\documentclass[aps,pra,showpacs,amssymb,nofootinbib,superscriptaddress,twocolumn, compress]{revtex4-1}
\usepackage{graphicx}

\usepackage[utf8]{inputenc}

\usepackage{amsmath}
\usepackage{amssymb}
\usepackage{amsthm}
\usepackage{mathtools}
\DeclarePairedDelimiter{\ceil}{\lceil}{\rceil}

\newtheorem{lemma}{Lemma}

\newtheorem{result}{Result}
\usepackage{dsfont}
\usepackage{csquotes}
\usepackage{ulem}
\usepackage{multirow}
\usepackage{color}

\newcommand{\id}{\mathds{1}}
\newcommand{\ii}{\mathrm{i}}

\usepackage{orcidlink}
\usepackage{bm, physics,dsfont}

\usepackage{hyperref}
\hypersetup{colorlinks=true,linkcolor=blue, linktocpage}

\begin{document}

\title{Selective and noise-resilient wave estimation with quantum sensor networks}

\author{Arne Hamann\,\orcidlink{0000-0002-9016-3641}}
\affiliation{Universit\"at Innsbruck, Institut f\"ur Theoretische Physik, Technikerstra{\ss}e 21a, 6020 Innsbruck, Austria}

\author{Paul Aigner\,\orcidlink{0009-0004-4277-2209}}
\affiliation{Universit\"at Innsbruck, Institut f\"ur Theoretische Physik, Technikerstra{\ss}e 21a, 6020 Innsbruck, Austria}

\author{Wolfgang D\"ur\,\orcidlink{0000-0002-0234-7425}}
\affiliation{Universit\"at Innsbruck, Institut f\"ur Theoretische Physik, Technikerstra{\ss}e 21a, 6020 Innsbruck, Austria}

\author{Pavel Sekatski\,\orcidlink{0000-0001-8455-020X}}
\affiliation{University of Geneva, Department of Applied Physics, 1211 Geneva, Switzerland}
\date{\today}

\begin{abstract}
    We consider the selective sensing of planar waves in the presence of noise. We present different methods to control the sensitivity of a quantum sensor network, which allow one to decouple it from arbitrarily selected waves while retaining sensitivity to the signal. Comparing these methods with classical (non-entangled) sensor networks we demonstrate two advantages.
    First, entanglement increases precision by enabling the Heisenberg scaling. Second, entanglement enables the elimination of correlated noise processes corresponding to waves with different propagation directions, by exploiting decoherence-free subspaces. 
    We then provide a theoretical and numerical analysis of the advantage offered by entangled quantum sensor networks, which is not specific to waves and can be of general interest. We demonstrate an exponential advantage in the regime where the number of sensor locations is comparable to the number of noise sources. 
    Finally, we outline a generalization to other waveforms, e.g., spherical harmonics and general time-dependent fields. 
\end{abstract}

\maketitle
\section{Introduction}
Waves are a fundamental concept in physics. Their detection is a crucial ingredient for many modern technologies and scientific experiments. Radios and smartphones use waves to transmit information. Radar uses waves to probe the environment~\cite{hulsmeyerVerfahrenUmEntfernte1905}, and telescopes detect waves to investigate their source~\cite{BurkeBernardF.2019AItR}. The importance of waves is not limited to the electromagnetic spectrum. Gravitational waves give a new perspective on the universe~\cite{ligoscientificcollaborationandvirgocollaborationObservationGravitationalWaves2016}, and seismic waves, as an example from the acoustic domain, probe earth's internal structure~\cite{2003PMSI}. At the fundamental level, the coupling of waves to matter is governed by quantum physics. Since waves are specific space-time dependent fields,  it is natural to consider quantum sensor networks for their detection. 

A quantum sensor network is an ensemble of local point-like sensors distributed in space and connected by a quantum network. The quantum network distributes entanglement, which can be used as a resource for the estimation of various space-dependent signals ~\cite{eldredgeOptimalSecureMeasurement2018,qianHeisenbergscalingMeasurementProtocol2019,geDistributedQuantumMetrology2018,sidhuGeometricPerspectiveQuantum2020,rubioQuantumSensingNetworks2020,qianOptimalMeasurementField2021,bringewattProtocolsEstimatingMultiple2021,zhangDistributedQuantumSensing2021,shettellPrivateNetworkParameter2022, qianOptimalMeasurementField2021,zhangDistributedQuantumSensing2021,bugalhoPrivateRobustStates2024,hassaniPrivacyNetworksQuantum2024,proctorMultiparameterEstimationNetworked2018,eldredgeOptimalSecureMeasurement2018,geDistributedQuantumMetrology2018,rubioQuantumSensingNetworks2020,sekatskiOptimalDistributedSensing2020,hamannApproximateDecoherenceFree2022,wolkNoisyDistributedSensing2020,hamannMultiparameter2024,BateExperimental2024}. Due to the presence of entanglement the sensitivity of quantum sensor networks can exhibit the so-called Heisenberg scaling~\cite{giovannettiAdvancesQuantumMetrology2011}, which provides a quadratic advantage in terms of the number of sensor locations, as compared to the standard quantum limit governing classical sensor networks. This scaling advantage is however fragile and can be quickly lost when noise is present.

In this paper, we specifically discuss how quantum sensor networks can be used for probing waves selectively. We consider a setting where the amplitudes of one or several plane waves play the role of the signal to be estimated while other waves are treated as noise. We then show how to design quantum sensor networks which are protected from noise while retaining the best possible sensitivity to the signal waves, maintaining Heisenberg scaling.

As already mentioned, noise significantly reduces estimation precision of sensor networks, and several techniques have been developed to overcome this effect. Quantum error-correction techniques, originally considered in quantum computing, have been adapted for the sensing tasks~\cite{PhysRevLett.112.080801,arradIncreasingSensingResolution2014,demkowicz-dobrzanskiAdaptiveQuantumMetrology2017a, sekatskiQuantumMetrologyFull2017,zhouAchievingHeisenbergLimit2018,faistTimeEnergyUncertaintyRelation2023}.
In general, error-correction for sensing requires fast universal control, which in the context of sensor networks means performing global entangled operations at a fast rate. For commuting spatially correlated noise error-correction offers no advantage~\cite{hamannMultiparameter2024} over an alternative less demanding technique based on decoherence-free subspaces~\cite{sekatskiOptimalDistributedSensing2020, wolkNoisyDistributedSensing2020, hamannApproximateDecoherenceFree2022, hamannMultiparameter2024,BateExperimental2024}. Here, the quantum sensor network is intrinsically protected from noise by preparing the physical sensors in a well chosen entangled state and applying local control (see Sec.~\ref{sec:dfs} for a formal definition).
Remarkably, in presence of correlated noise the advantage offered by quantum sensor networks is even more pronounced.
For certain examples it is known that the usage of decoherence-free subspaces can lead up to an exponential advantage in sensitivity of quantum sensor networks as compared to their classical counterparts~\cite{sekatskiOptimalDistributedSensing2020,hamannApproximateDecoherenceFree2022,wangExponentialEntanglementAdvantage2024}.

\begin{figure}[b]
    \centering
    \includegraphics[width=0.75\linewidth]{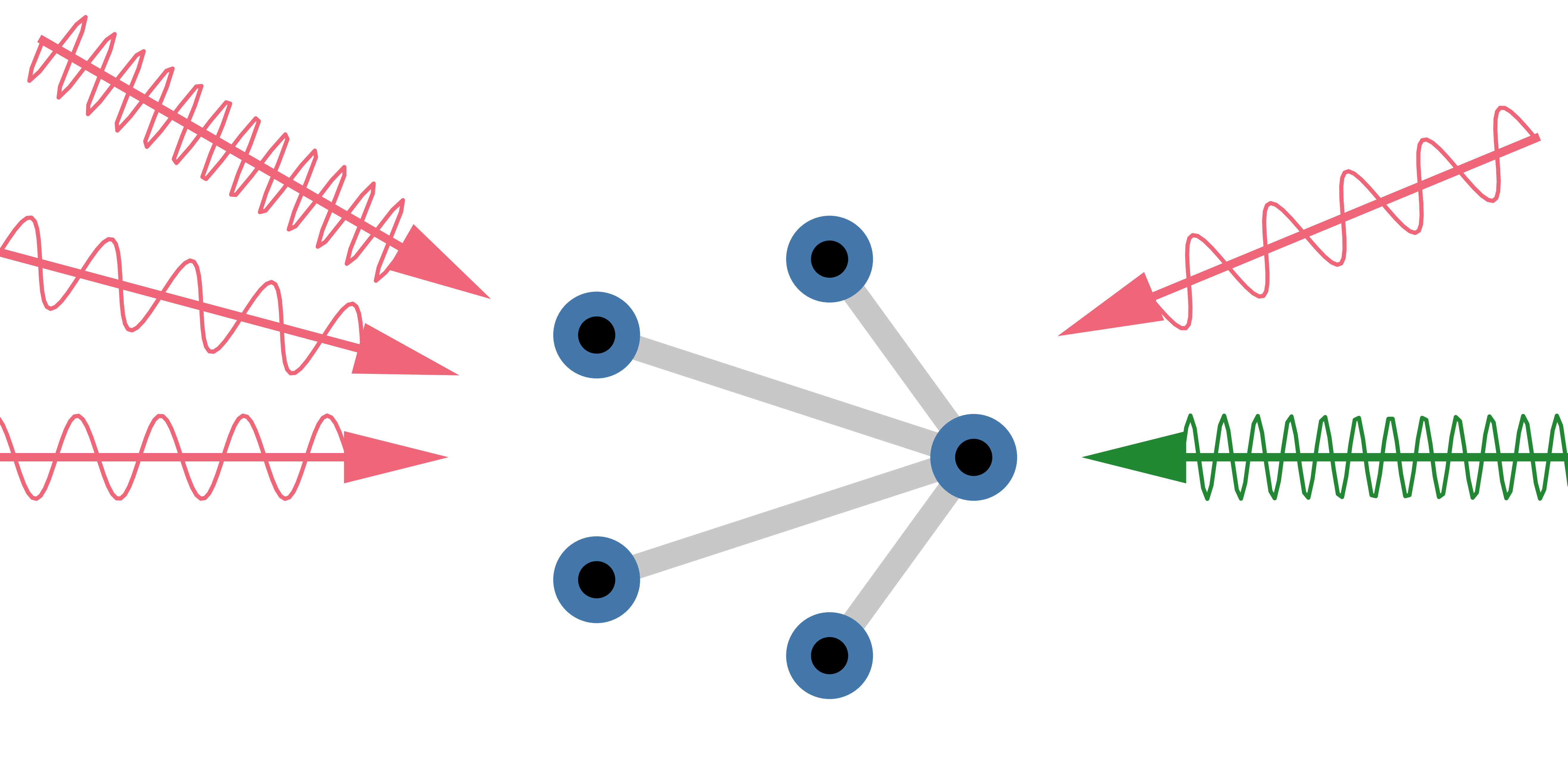}
    \caption{A network of quantum sensors (blue points) is used to estimate a signal wave (green) in the presence of noise waves (red), which may have different frequencies and different propagation directions. The sensor network can be tailored to any wave pattern by changing the applied control sequence.}
    \label{fig:set_up}
\end{figure}

 Here, we combine two control techniques: lock-in amplification, allowing one to filter the signals of the desired frequency, and decoherence-free subspaces (DFS), allowing one to filter the signals with the desired wavevectors.  The paper is structured as follows. In Sec.~\ref{sec: setup} we formally introduce the quantum sensor network setup described above. Then, in Sec.~\ref{sec: methods} we present the necessary background information on lock-in amplification, decoherence-free subspaces, and multi-parameter quantum metrology 
(in presence of correlated noise). Section~\ref{sec: results} presents the main results of the paper, namely:
\begin{itemize}
    \item A method to engineer control sequences and decoherence-free subspaces tailored to specific signals in the presence of specific sources of noise.
    \item A general theoretical analysis of the quantum advantages offered by entangled sensor arrays, which are not specific to waves and can be of independent interest.
\end{itemize} In Sec.~\ref{ref: examples} we illustrate our methods by applying them to various examples. Finally, Section~\ref{sec: generalization}
discusses possible generalization of these methods to other ``wavelets'' like spherical harmonics or signals with arbitrary fixed time dependence. Here we also show how one can increase the dimension of the DFS, which is beneficial in the context of Bayesian metrology. We finish with a conclusion in Sec.~\ref{sec:conslusion}.


\section{setup}
\label{sec: setup}

We will consider $d+1$ plane waves, described by fields of the form
\begin{equation}\label{eq:field}
f_j(\vec{x},t) = \cos\left({\vec{k}_j^\intercal\Vec{x} -\omega_j t +\phi_j}\right),
\end{equation}
coupled to a quantum sensor network of $n$ qubit sensors distributed at locations $\{\Vec{x}_1,\cdots,\Vec{x}_n\}$. It can be  initialized in an arbitrary state $\vert\psi\rangle$.\\

Each plane wave $j$ couples to the sensor qubits via the commuting local operators
\begin{align}\label{eq: W}
	\hat{W}_j(t)                 & =\sum_{i=1}^n f_j(\vec{x},t) \sigma_z^i,
\end{align}
where $\sigma_z^i$ is the Pauli operator acting the $i$-th qubit. For convenience we label the computational basis states with $z_i=\pm 1$ such that $\sigma_z^i\ket{z_i} = z_i \ket{z_i}$, and denote the elements of the product basis with $\ket{\bm z}$ for $\bm z = (z_1,\dots, z_n)$.
Generally, one can not assume knowledge of the phases $\phi_j$ in Eq.~\eqref{eq: W}, while it might be a valid assumption in other cases. Therefore, we will study both cases, if the phase is unknown it follows a flat prior distribution over the interval $[0,2\pi]$.

Additionally, we allow for bit-flip control, i.e., applying an $X$ gate locally to the sensors. This bit-flip gates are assumed instantaneous and effectively flip the sign of the local interaction. We will model this with an effective local interaction strength 
\begin{equation}
	C_i(t)=(-1)^{\eta_i(t)}
\end{equation}
where $\eta_i(t)$ is an integer counting the number of flips applied to the $i$-th sensor before time $t$. 
We talk of fast control if the flips can be applied with a frequency $\Omega\gg \omega_j$, which is fast enough so that the waves can be approximated as a constant between two flips. In this case, the control can be effectively approximated to take any integrable function $C_i(t)\in [-1,1]$ in the interval, which will be convenient for some calculations.
On the other hand, we talk of slow control if the frequency of flips $\Omega\approx \omega_j$ is of the order of the waves frequencies. In this case we stick with integer-valued $C_i(t)\in \{-1,+1\}$. \\

The amplitudes of each plane wave are given by $\beta_j$, so that the total Hamiltonian reads 
\begin{equation}
H(t) = \sum_{j=0}^{d} \beta_j \sum_{i=1}^n C_i(t) \, f_j(\vec{x}_i,t) \sigma_z^i,
\end{equation}
note that for simplicity we do not write the physical coupling constant of the field with the  sensor qubits, absorbing the later in the field amplitude.
Since all the operators in the Hamiltonian commute, the time evolution of duration $T$ results in the global unitary of the form

\begin{align}
U(T) &= \exp(- \ii \sum_{j=0}^n \beta_j \hat G_j) 
\\
\label{eq: G eig}
\hat G_j 
= \hat{G}(\omega_j,\vec{k}_j,\phi_j) 
&= \sum_{\bm z} g_{\bm{z}}(\omega_j,\vec{k}_j,\phi_j) \ket{\bm{z}}\bra{\bm{z}} \\
\label{eq: G int} 
 g_{\bm{z}}(\omega_j,\vec{k}_j,\phi_j)&= \int_{-\frac{T}{2}}^{\frac{T}{2}} \sum_{i=1}^n z_i f_j(\vec{x}_i,t) C_i(t) dt.
\end{align}
A generator $\hat{G}(\omega,\vec{k},\phi)$ with eigenvalues $g_{\bm{z}}(\omega,\vec{k},\phi)$, corresponds to a wave described by $\omega, \vec{k}$ and $\phi$, and is diagonal in the product basis $\ket{\bm z}$. 

For concreteness, let us consider the single-parameter estimation scenario where the parameter of interest $\theta=\beta_0$ is the amplitude of the first plane-wave. The other plane-waves are considered as noise. In particular, we focus on the scenario where their amplitudes $\bm \beta= (\beta_1,\dots,\beta_d)$ are distributed accordingly to some probability density $p(\bm \beta)$.
The overall time evolution 
\begin{equation}
    \rho(\theta) = \mathcal{N}[U_\theta\,  \rho_0 \, U_\theta^\dag] = U_\theta\,\mathcal{N}[  \rho_0 ]\, U_\theta^\dag
\end{equation}
can thus be decomposed as the unitary part $U_\theta =e^{-\ii \theta \hat G_0}$ encoding the parameter and the (commuting) completely positive trace preserving (CPTP) noise map
\begin{equation}\label{eq: noise channel}
\mathcal{N}_{\bm{\phi}}[\rho] = \int e^{-\ii \sum_{j=1}^d \beta_j \hat G_j}\, \rho \, e^{\ii \sum_{j=1}^d \beta_j \hat G_j} p(\bm \beta )\dd \bm \beta.
\end{equation}
In the case where the phases $\phi_j$ of the noise fields are unknown, one has to average over their values in addition. Then, the CPTP map becomes
\begin{equation}
\mathcal{N}[\rho] = \int \mathcal{N}_{\bm{\phi}}[\rho] \prod_{j=1}^d \frac{\dd \phi_j}{2\pi}.
\end{equation}

\section{Methods} 

\label{sec: methods}

In this section we introduce the necessary notation and review some results from the literature that we use later. In Sec.~\ref{sec: lock-in}
we discuss the lock-in amplification technique, which consists of applying a specific sequence of control pulses during the sensing, effectively decoupling the sensor from signals of unwanted frequency. In Sec.~\ref{sec:dfs} we introduce the formalism of DFS. The later is a subspace of the Hilbert space associated to the quantum sensor, which is invariant under the action of the noise generators and hence protected from noise. We then combine a protection method for correlated noise in time (lock-in amplification) with one for correlations in space (decoherence-free subspaces) to obtain decoherence-free subspaces for waves.
In Sec.~\ref{sec: metro} we review some results on noisy (multi)parameter, and specifically on the role of  DFS formalism in this context. In particular, we identify the coupling strength to waves $g_{\bm{z}}(\omega,\vec{k},\phi)$  as the relevant property and relate it to typical metrology concepts such as the mean squared error, the quantum Fisher information and the score function in Bayesian estimation.

\subsection{Fourier transform}
We now briefly discuss how to compute the eigenvalues of a general generator $\hat{G}(\omega,\vec{k},\phi)\ket{\bm z}$ in Eq.~\eqref{eq: G eig}. One notes that it can be conveniently expressed with the help of Fourier transforms of the control sequences as
\begin{equation}
    g_{\bm{z}}(\omega,\vec{k},\phi) = \Re\left(e^{i\phi} \sum_i z_i e^{i \vec k ^\intercal \vec{x}}\, \mathfrak{F}\left( C_i(t)\right)\right),
\end{equation}
where
\begin{equation}
	\mathfrak{F}\left(f(t)\right) = \int dt \ e^{- i\omega t}f(t).
\end{equation}
To ensure that the integration bounds here are consistent with Eq.~\eqref{eq: G int} and with the duration of the experiment, from now on we set $C_i(t)=0$ for $t\leq \frac{T}{2}$ and $t\geq \frac{T}{2}$

\subsection{Lock-in amplification}
\label{sec: lock-in}

\begin{figure*}
	\centering
    \includegraphics[width=0.7\linewidth]{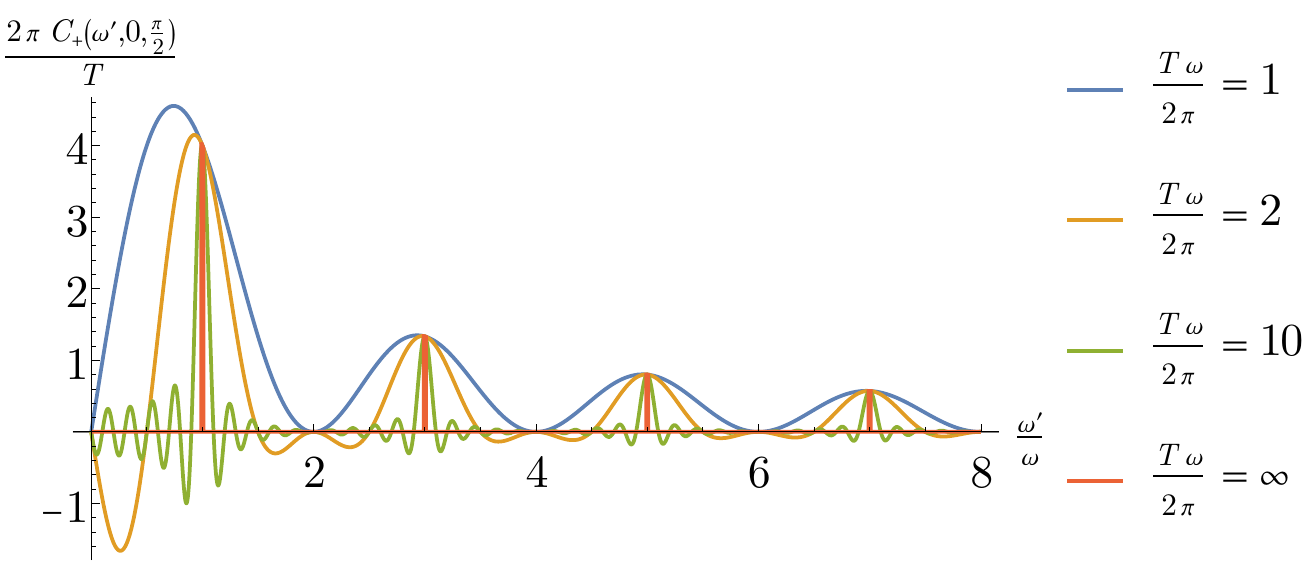}{}\\
	\includegraphics[width=0.7\linewidth]{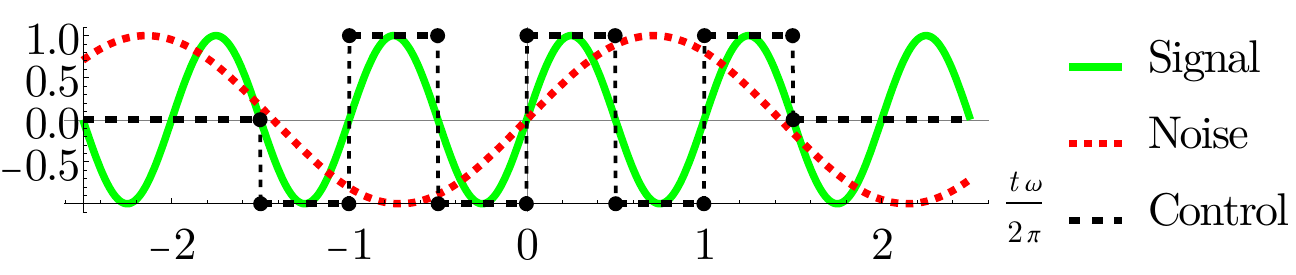}{}
    \caption{
 \textbf{Bottom}: Dynamical Decoupling is used to lock into a certain frequency. 
 Flipping the sensor to \enquote{align} with the signal makes the signal effectively positive. Therefore, the total phase accumulates over time. 
 Here the case of $\frac{T\omega}{2\pi}=3$ is shown.
 \textbf{Top}: The coupling to generic frequencies $\omega'$ for sensor locked into $\omega$. The qubit couples to $\omega$ and its odd harmonics $(2\mathbb{Z}+1)\omega$. The longer $T$, the sharper the coupling gets. For $\lim T\to\infty$, fields with generic frequencies are flipped each time at different phases, effectively accumulating to zero.
  }
	\label{fig:frequence_locking}
\end{figure*}

A typical technique in (quantum) control and metrology is lock-in amplification.
Let us recall this idea with the example of a a single qubit sensor located at position $\vec x=\vec{0}$.
To lock-in the qubit to an AC signal of frequency $\omega$, bit flip gates ($\pi/2$ pulses) are applied to the sensor qubit with frequency $2\omega$, as illustrated in Fig.~\ref{fig:frequence_locking}.
In our notation, this corresponds to the control sequence of the form 
\begin{align}
    C(t) & =
    \begin{cases}
\Pi_{\frac{\pi}{2}}(\omega t-\frac{\pi}{2}) & t\in \left[-\frac{T}{2},\frac{T}{2}\right]\\
0& t \notin \left[-\frac{T}{2},\frac{T}{2}\right]
\end{cases},
\end{align}
where  
\begin{equation}
    \Pi_{\gamma}(t) = \begin{cases}1 \quad \text{if} -\gamma < t + 2\pi \mathbb{Z} <\gamma \\ -1 \quad \text{else}\end{cases}
\end{equation}
is a rectangular wave.
To compute the resulting coupling factors $g_{\pm1}(\omega',\phi)$ to a wave of frequency $\omega'$, is in convenient to choose an anti-symmetric control (the $-\frac{\pi}{2}$ shift) and express it as  $C(t)=\sum_{l=1}^{\frac{T\omega}{2\pi}}\mathrm{rect}(\frac{\omega t}{\pi}+\frac{2l-1}{2})-\mathrm{rect}(\frac{\omega t}{\pi}-\frac{2l-1}{2})$. Then using  $\mathfrak{F}(\mathrm{rect}) = \mathrm{sinc}\left(\frac{\omega}{2}\right)$ and the linearity of the Fourier tansform one finds 
\begin{align}
    &g_{\pm 1}(\omega',\phi) = \mathrm{Re}\left( e^{\ii \phi} \mathfrak{F}\left(\pm C(t)\right)\right) =\\
    &= \mp\frac{2\pi }{\omega}\sin(\phi)\mathrm{sinc}\left(\frac{\pi}{2}\frac{\omega'}{\omega}\right)\sum_{l=1}^{\frac{T\omega}{2\pi}}(-1)^l\sin( \frac{(2l-1)\pi}{2} \frac{\omega'}{\omega}), \nonumber
\end{align}
depicted in Fig.~\ref{fig:frequence_locking}.
In the limit of $T\to\infty$
\begin{align}
    &g_{\pm 1}(\omega',\phi) =\\ 
     &=\mp\frac{2\pi }{\omega}\sin(\phi)\mathrm{sinc}\left(\frac{\pi}{2}\frac{\omega'}{\omega}\right)\sum_{l=1}^{\infty}(-1)^l\delta\left( \frac{\omega'}{\omega}-(2l-1) \right)
\end{align}
it only couples to $\omega$ and its harmonics. An asymptotic behavior that only couples to $\omega'$ and not to it's harmonics can be obtained with the control sequence given by $C(t)=\sin(\omega t)$ instead of the rectangular function, but it comes at the cost of a smaller prefactor (residual sensitivity to the AC signal).


\subsection{Decoherence-free subspaces for waves}\label{sec:dfs}
Another known technique to decouple from spatially correlated noises, is to prepare the state inside a decoherence-free subspaces (DFS)~\cite{lidar2014review}. It has been studied in the context of distributed sensing~\cite{sekatskiOptimalDistributedSensing2020, wolkNoisyDistributedSensing2020, hamannApproximateDecoherenceFree2022, hamannMultiparameter2024}, and can be used in combination with lock-in amplification. Here we discuss the decoherence-free subspaces for the specific case of waves.

A state $\ket{\psi}= \sum_{\bm z} c_{\bm z} \ket{\bm z}$ is protected from a wave generator $\hat{G}(\omega,\vec{k},\phi)$ if all the states $\ket{\bm z}$ on which it is supported have the same eigenvalue (coupling strength).
Formally, given $d$ such generators we define a so called affine DFS 
\begin{equation}\label{eq: DFS def}
	\mathrm{DFS}_{\bm \kappa} = \left\{\bm z \Big\vert g_{\bm z}(\omega_j, \vec{k}_j, \phi_j) = \kappa_j \quad \forall 1\leq j \leq d \right\}.
\end{equation}
as the set of all vectors $\bm z$ with fixed couplings $\bm \kappa = (\kappa_1,\dots,\kappa_d)$ to all the waves.

An important special case, that we simply call the (standard) decoherence-free subspace and write 
\begin{equation}
\mathrm{DFS}=\mathrm{DFS}_{\bm{\kappa}=0},
\end{equation} 
is where all the couplings are zero.
Notice that if $\bm z$ is inside the DFS, so is $-\bm z$, therefore the GHZ-type state $\ket{\phi_{\bm z}}=\frac{\ket{\bm{z}}+\ket{-\bm z}}{\sqrt{z}}$ are also inside the DFS. Such states are useful resources for noise-protected sensing. In particular, they are optimal in the single-parameter Fisher regime~\cite{sekatskiOptimalDistributedSensing2020}, and optimal up to a factor of at most $4$ in the multiparameter Fisher regime \cite{hamannMultiparameter2024}.


In the limit of strong noise~\cite{hamannMultiparameter2024}, i.e. where  the variance of $p(\bm{\beta})$ in Eq.~\eqref{eq: noise channel} is very large, the noise acts as a twirling map projecting the initial state onto all the affine DFS. Formally, we can write 
\begin{equation}\label{eq: noise strong}
    \mathcal{N}[\rho_0] = \sum_{\mathrm{DFS}_{\bm \kappa}}\, \Pi_{\bm\kappa} \, \rho_0 \, \Pi_{\bm\kappa} =\bigoplus_{\bm \kappa} p_{\bm \kappa} \rho_{\bm \kappa}
\end{equation}
with $\Pi_{\bm\kappa} =\sum_{\bm z\in \rm{DFS}_{\bm \kappa}} \ketbra{\bm z}$,  $p_{\bm \kappa} = \tr(P_{{\bm \kappa}} \rho_0)$, and $\rho_{\bm \kappa}=\frac{\Pi_{\bm\kappa} \rho_0 \Pi_{\bm\kappa}}{p_{\bm \kappa}}$. Since the signal commutes with the noise, the final state of the sensors $\rho(\theta) = \bigoplus_{\bm \kappa} p_{\bm \kappa} \rho_{\bm \kappa}(\theta)$ is then also a mixture of states over different affine DFS. Since sensing problems are, in general, convex, it is already clear that the optimal strategy here is to prepare the initial state inside some affine ${\rm DFS}_{\bm \kappa}$ (often the standard one).\\


In certain cases, it can happen that if a state is perfectly protected from noise, it is also decoupled from the signal. Then, it can be very useful to introduce the notion of approximate decoherence-free subspaces~\cite{hamannApproximateDecoherenceFree2022}; see the example Section~\ref{sec:example:approx_dfs}. In all cases, for our purpose here, the decoherence-free subspaces can be analyzed and understood by looking at the coupling strengths to waves $g_{\bm z}(\omega, \vec{k}, \phi)$.

\subsection{Connection to Metrology}
\label{sec: metro}
Quantum metrology is often studied in the so-called Fisher regime, where the parameter is estimated after many repetitions of the experiment. Here, the quantum Fisher information (matrix) QFI(M) is the central figure of merit. The QFI(M) lower bounds the mean square error of any unbiased estimator $\bm{\hat{\theta}}$ of the parameters $\bm{\theta}$ via the (quantum) Cramer-Rao bound
\begin{equation}
    \mathrm{MSE}(\bm{\hat{\theta}}) \geq \mathrm{COV}(\bm{\hat{\theta}})\geq \mathcal{F}^{-1}.
\end{equation}

For a single parameter, the optimal state in the noiseless case is an equal superposition of an eigenstate with minimal and maximal energy with respect to the generator. As already mentioned, in the limit of strong noise, it is known~\cite{sekatskiOptimalDistributedSensing2020} that the optimal strategy is to prepare the GHZ-like state $\ket{\mathrm{GHZ}_{\bm z}} = \frac{1}{\sqrt{2}}\left(\ket{\bm z}+\ket{-\bm z}\right)$ from the DFS (with best sensitivity), and perform local measuremnts.
In particular, with one DFS the QFI of the state $\ket{\psi} = \sum_{{\bm{z}} \in \mathrm{DFS}_{\bm{\kappa}}} c_{\bm{z}} \ket{\bm z}$ with respect to the signal generated by $\hat{G}(\omega,\vec{k},\phi)$ simply reads
\begin{equation}
    \mathcal{F}= 4 \sum_{\bm z} |c_{\bm z}|^2 g_{\bm z}(\omega,\vec{k},\phi)^2 -  4 \left(\sum_{\bm z} |c_{\bm z}|^2 g_{\bm z}(\omega,\vec{k},\phi)\right)^2,
\end{equation} 
see appendix~\ref{appendix:qfi} for derivation. Hence, maximizing the sensitivity amounts to find the minimal and maximal values of $g_{\bm z}(\omega,\vec{k},\phi)$ within the DFS.

It is important to emphasize that in this strategy the optimal measurement on the final state is always local, i.e. entanglement is only used for state preparation. Of course in general this is not the case, in particular we will see in sec.~\ref{sec:numeric} that for product initial states the final QFI can drasticaly overestimate the best Fisher information obtainable with product measurements.

In the multi-parameter case, it was shown that a strategy where different GHZ-type states from the DFS are prepared sequentially (and measured locally) is optimal up to a factor $4$ at most. Note that, here optimality is with respect to the partial-order of the Fisher information matrices describing the strategies~\cite{hamannMultiparameter2024}.  In the appendix~\ref{appendix:qfi} we give the form of them QFIM for multi-parameter estimation, the QFI for GHZ states in approximate DFS and the QFI for separable states in the noiseless case.
\\

Another well established framework is the Bayesian metrology. Here, one assumes a prior knowledge of the parameter $\bm{\theta}$ given by a probability distribution $p(\bm{\theta})$, and computes the post-measurement distribution using the Bayes rule 
$ p(\bm{\theta}|x) = \frac{p(x|\bm \theta) p(\bm \theta)}{p(x)},$ where $x$ is the observed measurement result. Given a score function ${\rm S}_x ={\rm  Score }(p(\bm \theta|x))$  quantifying how good the updated distribution is (e.g. the MSE of an estimator $\hat{\bm \theta}_x$), one is looking for the strategy maximizing the expected score $\bar{\rm S}= \sum_x p(x) {\rm S}_x$. As already pointed out, in the large noise limit the optimal strategy is always to prepare the initial state within a single ${\rm DFS}_{\rm \kappa}$. The problem of finding the best states in the Bayesian regime has been studied in \cite{vasilyevOptimalMultiparameterMetrology2024, PhysRevResearch.6.023305}. In contrast to the Fisher regime, here, two-dimensional DFS and GHZ-type states are generally not optimal, as the later can at most encode one bit of information~\cite{wolkNoisyDistributedSensing2020}.






In all cases, the optimal sensing strategy for strong noise is to prepare the state within one DFS. The tools for finding the optimal sensor within this DFS, i.e. the noiseless case, are well established, see e.g.~\cite{giovannettiAdvancesQuantumMetrology2011,vasilyevOptimalMultiparameterMetrology2024,goreckiOptimalProbesErrorcorrection2020a}, and can be applied at this stage.\\

It is worth mentioning that even in the case where the noise is weak, and does not completely destroy the coherence between different DFS,
it is beneficial to find states whose coupling to the signal waves is stronger than the coupling to the noise waves. In this case, approximate DFSs formalism can be used to increase this signal-to-noise ratio, see example~\ref{sec:example:approx_dfs}.\\

 To summarise, the metrology aspect of the problem is well understood. Hence, we now focus on the task of designing the DFS in the first place.\\

\section{Results}
\label{sec: results}

\subsection{Construction of a DFS for given sensor positions}\label{sec:DFS_control}
If the noise and the signal have different frequencies, they can be distinguished in a single location. Lock-in amplification is sufficient to cancel the noise, and quantum sensor networks are not required to decouple from the noise.
In Appendix~\ref{sec:wave_locking}, we discuss how this method can be applyed to a quantum network such that a high signal coupling is maintained. This section will focus on cases where the signal and noise have the same frequency and, hence, are locally indistinguishable. In this case, entanglement is required to protect the sensor network from the noise. 
Then, it is reasonable to assume that the interaction time $T$ is a multiple of the period.

\begin{result}
\textit{The DFS that protects $\ket{\bm{z}}$ and $\ket{\bm{-z}}$  from $d$ waves with equal frequency $\omega$ for $n$ given sensor positions is constructed following a three-step method:
\begin{enumerate}
	\item Determine the field matrix $F$ 
	\item Orthogonalize $F$
	\item Identify the control sequence
\end{enumerate}
For known (unknown) wave phases, the protocol is guaranteed to be successful if $n>d$ ($n>2d$). For point symmetric sensors, the phases can be considered to be known, and $n>d$ is sufficient.
}
\end{result}

\subsubsection{Determine the field matrix $F$}
If the noise phases are known, the field matrix is an $d+1$ times $n$ function-valued matrix with the elements
\begin{equation}
    F_{ji}=\cos(\vec{k}_j^\intercal\vec{x}_i - \omega t +\phi_j).
\end{equation}
The first $d$ rows correspond to the noise fields, and the last row corresponds to the signal field.

For unknown noise phases, notice that 
\begin{align}\label{equ:unknown_phase}
	\cos(\vec{k}_j^\intercal\vec{x}_i - \omega t +\phi_j) &= 
	 A_j\cos(\vec{k}_j^\intercal\vec{x}_i - \omega t) \\
	&+ B_j\sin(\vec{k}_j^\intercal\vec{x}_i - \omega t)\nonumber
\end{align}
can be decomposed in a symmetric and an anti-symmetric component, where the coefficients $A_j$ and $B_j$ are independent of the sensor $i$. Hence a plane wave with unknown phase $\phi$ is in the subspace spanned by $\cos(\vec{k}_j^\intercal\vec{x}_i - \omega t)$ and $\sin(\vec{k}_j^\intercal\vec{x}_i - \omega t)$.
We can now define $F$ similarly to the case with known phases, with just $2d+1$ rows containing cosine and sine components for each noise generator. This way a scenario with unknown phases reduces to a setup with twice as many noise generators.

\subsubsection{Orthogonalize $F$}
To construct a DFS that contains the states $\ket{\bm z}$ and $\ket{-\bm z}$,
use the Gram–Schmidt algorithms with rows of $F$ as vectors and with the scalar product 
\begin{equation}\label{equ:scalar-product}
	\langle x, y\rangle = \sum_{i=1}^n z_i \int_{-\frac{T}{2}}^{\frac{T}{2}}  x_i(t) y_i(t) dt
\end{equation} to orthogonalize $F$. Notice that this scalar product is chosen to reproduce the eigenvalues in Eq.~\eqref{eq: G int}, for certain input vectors.
After orthogonalization, the last row 
\begin{equation}
\bm{s}^\perp= \big(s_1^{\perp}(t),\dots, s_n^{\perp}(t)\big)
\end{equation}
is the orthogonal signal component, i.e. the component of the signal $(F_{(d+1)1},\dots,F_{(d+1)n})$ orthogonal to all the noises. 

Notice that $\bm{s}^\perp$ is only non-zero if the signal field restricted on the sensor locations is linearly independent from the noise. For this independence, it is necessary that the number of sensor locations is $n>d$ for the case of known phases and $n>2d$ for unknown phases (recall that $d$ is the number of noise sources).

\subsubsection{Identify the control sequence}
\label{eq: time control}
\paragraph{For fast control,} i.e., when $C^{fast}_i(t)$ can effectively be any integrable function between $-1$ and $1$, we can directly use a normalized orthogonal signal component as a control sequence
\begin{equation}
    C^{fast}_{i}(t) = \frac{s^\perp_i(t)}{\max_{i'}\sup_{-\frac{T}{2}\leq t\leq \frac{T}{2}} |s^{\perp}_{i'}(t)|}.
\end{equation}
The normalization comes from the fact that the absolute value of the control can be at most one. Additionally, we want to saturate this value, as this increases the residual signal and therefore the sensitivity.

The states $\ket{\bm z}$ and $\ket{-\bm z}$ will be contained in the DFS (see Eq.~\eqref{eq: DFS def}), as the vector of the control sequence is orthogonal to the noise fields $F_j$
\begin{equation}
    0 = \langle \bm{C}^{fast} , F_j \rangle = g_{\bm{z}}(\omega_j, \vec{k}_j, \phi_j) \quad \forall j, 
\end{equation}
and the scalar product in Eq.~\eqref{equ:scalar-product} was precisely chosen such that $\langle \bm{C}^{fast} , F_j \rangle$ gives the eigenvalue of the noise generator $\hat G_j$ in equation~\eqref{eq: G int}.

Notice that each element  of the orthogonal signal component $\bm s^\perp$ is a linear combination of the local fields, so it can be expanded as
\begin{equation}
    s^\perp_i =\sum_{j=1}^{d+1} A_{ij} \cos(-\omega t +\varphi_{ij})= A'_i \cos(-\omega t +\varphi_i)
\end{equation}
with trigonometric functions with equal frequency and different phases. Due to the harmonic addition theorem \cite{weisstein}, the control sequence at each sensor is a phase-shifted harmonic oscillation. 
Therefore the control of each sensor is given by  
\begin{equation}
    C^{fast}_i = A_i \cos(wt + \varphi_i) \propto s^\perp_i(t),
\end{equation}
with $A_i= \frac{ A_i'  }{\max_i \vert A'_i\vert}$ being the normalized amplitudes. 

\paragraph{For slow control,}
we will use the periodic function
\begin{equation}\label{eq: Pi cont}
    \Pi_{\gamma}(t) = \begin{cases}1 \quad \text{if} -\gamma < t + 2\pi \mathbb{Z} <\gamma \\ -1 \quad \text{else}\end{cases}
\end{equation}
which only requires two flips per period. Now the fast control sequence is replaced with 
\begin{equation}
    C^{slow}_i = \Pi_{\sin^{-1}{A_i}}(\omega t + \varphi_i).
\end{equation}
The slow control strategy will be in the decoherence-free subspace as the overlap of the fast and the slow control
\begin{equation}
	\int_{-\frac{T}{2}}^\frac{T}{2} C^{slow}_i(t) F_{ji}(t) dt =\frac{4}{\pi} \int_{-\frac{T}{2}}^\frac{T}{2}  C^{fast}_i(t) F_{ji}(t) dt 
\end{equation}
with the local oscillations $F_{ji} = \cos(\vec{k}_j^\intercal \vec{x}_i -\omega t +\phi_j)$ just differ by a constant (most importantly $i$ independent) value $\frac{4}{\pi}$. The detailed computation can be found in the appendix~\ref{appendix:slow_control}.
Hence, with the slow control the risidual coupling to the signal is  $\frac{4}{\pi}$ times better than for the fast control. Notice however that with the slow control the sensor network is not decoupled from noise fields with harmonic frequencies.

\subsubsection{Point symmetric sensors}
\label{sec:point-sym}

We say that a sensor network with the initial state $\vert \psi\rangle$ is point symmetric with respect to the point $\vec{p}$ if for each sensor position $\vec{x}_i$ there exists another sensor at position $\vec{x}_{i'}$ such that 
\begin{equation}
\vec{x}_i-\vec{p} = -\vec{x}_{i'}+\vec{p} \quad \text{ and } \quad z_i = z_{i'}
\end{equation}
for all $\bm z$ such that $\braket{\bm z}{\psi}\neq 0$.

For a sensor network symmetric around a point $\vec{p}$ it is natural to set this point as the origin of space coordinates. Then, one can easily see that the sensor is only sensitive to the symmetric component of the waves (with $f(\vec x,t )= f(-\vec x,t )$). This guarantees protection from all the sine components of the noise fields in Eq.~\eqref{equ:unknown_phase}. 
Therefore, we can construct the DFS by only considering the symmetric (cosine) noise components and using the symmetric part of the signal as the last row of $F$. The resulting control sequence $C_i^\perp(t)$ will be point symmetric as a linear combination of point symmetric functions. Thus, a point symmetric sensor network only requires $n>d$ sensors to be protected from $d$ noise sources, even if the phases are unknown.

\subsection{Construction of a DFS by sensor placement}\label{sec:no-control}
\begin{figure}
	\centering
	\includegraphics[width=0.5\linewidth]{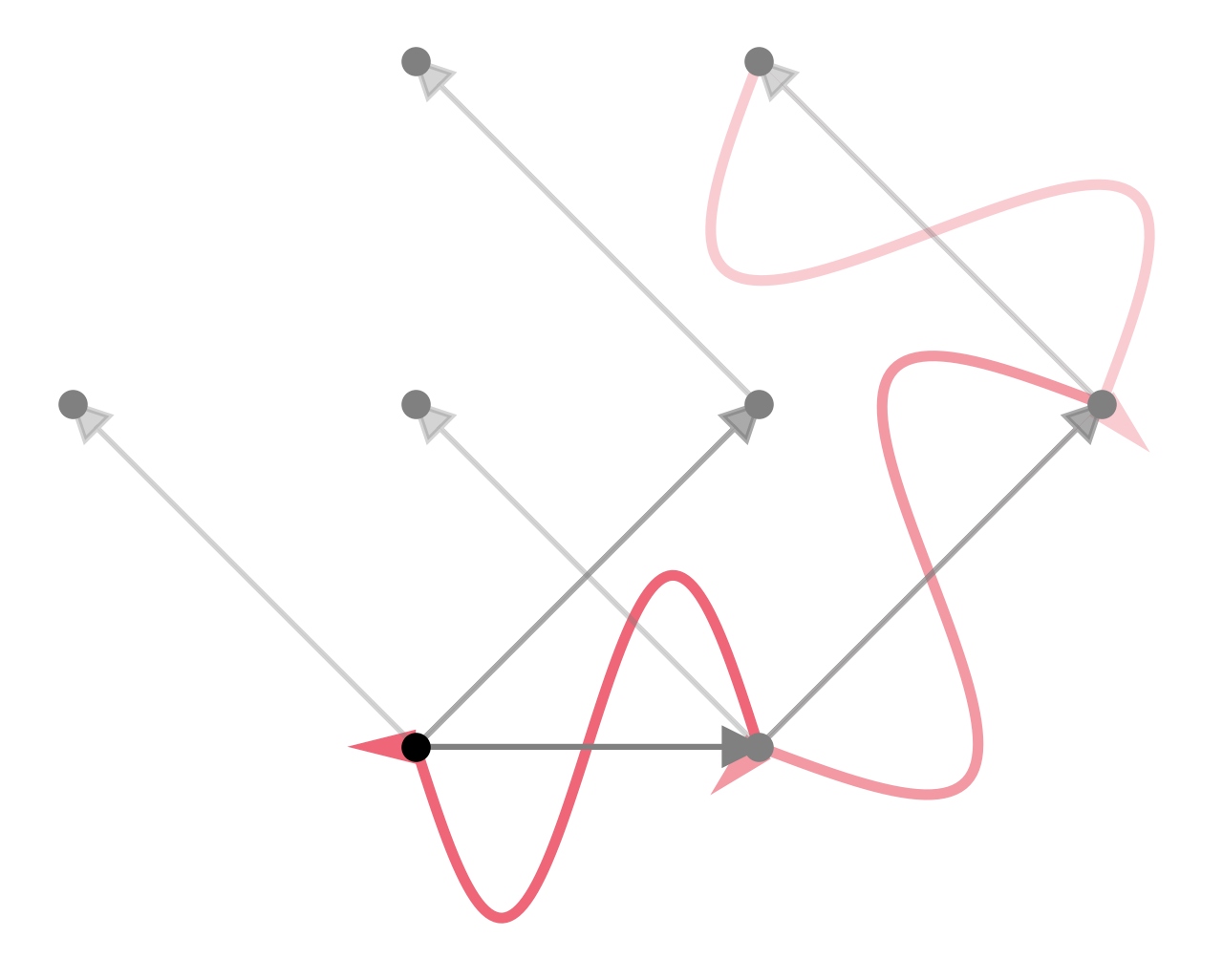}

	\caption{Constructing a sensor protected from 3 waves is done in 3 steps. First, the initial sensor (black dot) is shifted a wavelength in the travel direction of the first wave. Then, these two new sensor locations are shifted a wavelength along the second wave. Last, we shift all four sensors a wavelength in the direction of the third wave. Thereby, we get $2^3=8$ sensor locations, such that the effect of each wave cancels. }
	\label{fig:sensor_placement}
\end{figure}

We now discuss a method that aims to cancel $d$ noise waves exactly, given the freedom to arrange the sensors in space, as illustrated in Fig.~\ref{fig:sensor_placement}. We start with one sensor at an arbitrary position $\Vec{x}_1$. Then, for each new noise wave $f_j(\vec{x},t)$ added, one doubles the number of sensors by placing new ones at a wavelength in the traveling direction. Overall, we obtain a network with $2^d$ sensor locations, such that for every noise wave, we get pairs of sensor locations with opposite field directions. Hence, for both states $\ket{1}^{\otimes n}$ and $\ket{-1}^{\otimes n}$ the contributions of the noise wave will be canceled, and they define a two-dimensional DFS. 
To retain sensitive to the signal the state $\frac{\ket{1}^{\otimes n} + \ket{-1}^{\otimes n}}{\sqrt{2}}$ is prepared and all sensors are flipped with twice the signal frequency and no relative phase (lock-in amplification). Notice that the method will work for all noise frequencies and is independent of the noise phases. The signal will be lost if indistinguishable from the noises on the sensor positions.

\subsection{Comparison with non-entangled strategies}\label{sec:product}

We have now discussed a few techniques to devise noise-resilient sensing strategies. These strategies require to prepare the sensors in entangled states (recall that the measurements can be local), which can be challenging in practice. It is thus natural to ask if the advantages offered by these strategies are worth this effort. In other words compare the entangled strategies which those using a product state of the sensors (and local measurements, see below). In general, this question is very challenging, since showing a clear cut quantum advantage requires to optimize over all possible product strategies. 
Nevertheless, in some cases general results have been derived. The conclusions depend drastically on the considered estimation task, as we now discuss.

First, let us recall that in the noiseless case when estimating a single parameter, like a global field, classical sensor networks are fundamentally limited by the so-called standard quantum limit,  while quantum sensors are limited by the Heisenberg limit~\cite{giovannetti2011advances}. This suggest that a quadratic improvement in precision is possible in the noiseless case, and this advantage survives when estimating a single signal within a DFS.


In contrast, in the special case where there is no noise and the goal is to estimate the values of the field at all sensor locations, it has been shown that entanglement is not bringing any advantage~\cite{proctorMultiparameterEstimationNetworked2018}. A similar heuristic conclusion has been reached in \cite{PhysRevX.14.011033}, where the task was to estimate the pairwise difference of field values at all pairs of sensor locations in the presence of a constant filed noise. In fact, one can understand this task as estimating ``almost'' the entire vector of field values, with the exception of the sum of these values, which corresponds to the noise generator. Both results, thus, apply to the setting of low-rank (or no) noise and an ``almost full'' estimation of the field values vector, which seems to offer no room for significant quantum advantage.

In contrast, \cite{sekatskiOptimalDistributedSensing2020} considered the somehow opposite setting,  where $n$ sensors are used to estimate a single parameter (constant field) against $n-1$ correlated noise sources (standing waves with different wavelength). Here, it was shown that the DFS-based entangled strategies offer an exponential advantage (in $n$) over all product-state strategies.

 We see from these examples that depending on the estimation problem, the impact of entanglement ranges from no advantage at all to an exponential one. In the next subsection, we give some elements of response to where this advantage could be found for a generic estimation task. The following section~\ref{sec:product_general} makes some general remarks on product strategies, works out criteria for proven exponential improvements and summarizes known scenarios where these criteria are fulfilled. In section~\ref{sec:numeric} we compare the entangled strategy to particular product state strategies which seem good (without claim of their optimality).

\subsubsection{General arguments on product state strategies}\label{sec:product_general}

To compare the entangled and product-state strategies, we consider the worst-case setting, where each noise signal can be treated as a random variable with large enough fluctuations. In this case, as discussed in section~\ref{sec:dfs}, the effect of noise can be modelled by applying the dephasing map $\mathcal{N}$ in Eq.~\eqref{eq: noise strong}. This map commutes with the unitary encoding of the signal, and can thus be seen as random projection of the initial onto all the DFS 

\begin{equation}
\tilde \rho_0= \mathcal{N}[\Psi_0] = \bigoplus_{{\rm DFS}_{\bm \kappa}} \Pi_{\bm \kappa} \Psi_0 \Pi_{\bm \kappa}, 
\end{equation}
where $\Pi_{\bm  \kappa}$ is a projector onto the corresponding DFS, and $p_{\bm \kappa} =\tr \left( \Psi_0 \Pi_{\bm \kappa}\right)$ is the probability that the pure state $\Psi_0$ is projected there. Moreover, for any signal generator $G$, the QFI of the final state also decomposes accordingly 
\begin{align} \label{equ:FischerdecDFS}
    \mathcal{F}_G &= 4 \sum_{\bm \kappa} p_{\bm \kappa}  {\rm Var}_{\Psi_k} (G_{\bm \kappa}) 
\end{align}
where $\Psi_{\bm \kappa} = \frac{1}{p_{\bm \kappa}} \Pi_{\bm \kappa} \Psi_0 \Pi_{\bm \kappa}$, $G_{\bm \kappa} = \Pi_{\bm \kappa} G \Pi_{\bm \kappa}$, which is explicitly shown in the appendix \ref{Proof:equ:FischerdecDFS}. The general idea, is that in sensing tasks which exhibit quantum advantage there is only a limited number of DFS that are useful for sensing. Hence one has to design states which project onto these DFS with hight probability. As will show below, for any set of DFS (defined by the control strategy) and a product initial state $\rho_0$ one is very restricted in the control of the distribution $p_{\bm \kappa}$ of the subspaces on which the state is projected. \\

Clearly to make such a claim, the sensing task must be nontrivial enough. For instance, if one can protect from noise by simply placing one (or all sensors) at some position, the construction of the DFS becomes obsolete.
To avoid such degenerate cases and quantify the intrinsic complexity of the ``noise-protection'' task, we define  $m$ as the minimal number of sensors such that a DFS exists. Section~\ref{sec:DFS_control} upper-bounds $m\leq d+1$ by the number of different noise generators $d$, in the case of bit-flip control. Furthermore, it directly implies the following relation on the states $\ket{\bm z}$ belonging to the same DFS. 

\begin{lemma}
    {\it Consider a sensor array with $n$ sensors subject to noise. Let $m \leq n$ be the size of the smallest subset of sensors, for which a {\rm DFS} exists. Let $\bm z=(z_1, \dots,z_n),\bm z'=(z_1',\dots,z_n')\in {\rm DFS}$ with $z_i,z_i'\in\{-1,1\}$ be two vector from some DFS (not necessarily the ``minimal one''), then
\begin{equation}
    D_H(\bm z,\bm z'):=\sum_{i=1}^n\frac{1}{2}| z_i - z_i'| \geq m,
\end{equation}
where $D_H$ is the Hamming distance between the strings.}
\end{lemma}
\begin{proof}
    The proof is done by contradiction.
Assume $\bm z,\bm z' \in {\rm DFS}$ and $D_H(\bm z,\bm z')=m'<m$.
Without loss of generality we assume $\bm{z}=\bm{1}_{m'}\oplus \bm{1}_{n-m'}$ and $\bm{z'} = \bm{0}_{m'} \oplus \bm{1}_{n-m'}$, where $\bm{0}_x$ ($\bm{1}_x$) is the constant $0$ ($1$) tuple of length $x$.
Then we find
\begin{equation}
    (\ket{\bm{1}_{m'}}+\ket{\bm{0}_{m'}})/\sqrt{2} =  \tr_{n-m'}\left((\ket{\bm{z}}+\ket{\bm{z'}})/\sqrt{2}\right),
\end{equation}
where $\tr_{n-m'}$ is the partial trace over $n-m'$ qubits with identical eigenstate $z_i=z_i'$. Let $\Pi=\sum_{\bm z\in{\rm DFS}}\ketbra{\bm z} $ be the projector on to DFS
\begin{equation}
    (\ket{\bm{1}_{m'}}+\ket{\bm{0}_{m'}})/\sqrt{2} = \tr_{n-m'}\left(\Pi(\ket{\bm{z}}+\ket{\bm{z'}})/\sqrt{2}\right)
\end{equation}
and $\Tilde{\Pi} = (\id \otimes \bra{\bm{1}_{n-m'}})\Pi (\id \otimes \ket{\bm{1}_{n-m'}})=\sum_{\Tilde{\bm{z}}\oplus\bm{1}_{n-m'}\in{\rm DFS}}\ketbra{ \tilde{\bm z}} $, be effective reduction onto the first $m'$ sensors.
\begin{equation}
    (\ket{\bm{1}_{m'}}+\ket{\bm{0}_{m'}})/\sqrt{2} = \Tilde{\Pi}(\ket{\bm{1}_{m'}}+\ket{\bm{0}_{m'}})/\sqrt{2}
\end{equation}
Notice that $\Pi$ and $\Tilde{\Pi}$ are diagonal in the $\sigma_z$ basis. Hence, there exists a two dimensional $ \widetilde{\mathrm{DFS}}$ spanned by $\ket{\bm{1}_{m'}}$ and $\ket{\bm{0}_{m'}}$ on $m'$ qubits. This contradicts the assumption that $m$ is minimal.
\end{proof}

The projectors $\Pi_{\bm \kappa} =\sum_{\bm z\in{\rm DFS}_\kappa} \ketbra{\bm z}$ onto the different DFS are diagonal in the computational basis. Hence, all product computational states are entirely within a single DFS, and one can, in principle, choose a state such that $p_\kappa=1$ for any ${\rm DFS}_{\bm \kappa}$. These states are, however, also eigenstates of the signal generators and are completely useless for sensing. Thus to contribute to sensing, a product state $\ket{\Psi}$ must project onto several computational states within the same DFS.  This explains the motivation of the following result.
\begin{result} 
    \label{result:single_dfs_exponantially_small} {\it Consider a product state $\ket{\Psi}$, and a {\rm DFS} containing $k$ vectors $\bm z_1,\dots,\bm z_k$ at Hamming distance at least $m$. We arrange the vectors in order of decreasing probability, i.e. such that $p_i=|\braket{\bm z_i}{\Psi}|^2\geq p_{i+1}=|\braket{\bm z_{i+1}}{\Psi}|^2$. Then the following bound holds}
    \begin{equation} \label{equ: Boundprodprobdfs}
        \sum_{i=2}^k p_i \leq \frac{1}{\sum_{\ell=0}^{\lfloor \frac{m-1}{2}\rfloor} \binom{m}{\ell}}\leq \left(\frac{ m}{\lfloor \frac{m-1}{2}\rfloor}\right)^{-m}  \leq 2^{-m}\\
    \end{equation}
    It follows that the contribution from this DFS to the QFI  $\mathcal{F}_G$ in Eq.~\eqref{equ:FischerdecDFS} is bounded by
    \begin{equation}
    4\, p_{\bm \kappa} {\rm Var}_{\Psi_k} (G_{\bm \kappa})  \leq \| G_{\bm \kappa}  \|_{\rm Spec}^2\, 2^{-m+1},
    \end{equation}
    where $G_{\bm k} = \Pi_{\bm \kappa} G \Pi_{\bm \kappa}$ is the restriction of the signal generator on the DFS (with $\Pi_{\bm \kappa} = \sum_{i=1}^k  \ketbra{\bm z_i}$), and $\|G_{\bm k} \|_{\rm Spec}$ is the difference between its maximal and minimal eigenvalues (a pseudo-norm on the vector space of Hermitian operators).
\end{result}
The proof of the  inequality~\eqref{equ: Boundprodprobdfs} is given in the appendix~\ref{section:single_dfs_exponantially_small}. Intuitively one can think of this result as follows: either the state $\ket{\Psi}$ is (almost) product in the computation basis and useless for sensing, or the probability that it projects onto a single DFS is exponentially small in $m$. In both cases however the contribution to the QFI is exponentially small, since one can show that
\begin{align} \label{equ:fisher_info_term_ineq}
 4 \, p_{\bm \kappa} {\rm Var}_{\Psi_{\bm \kappa}} (G_\kappa) &\leq 2 \| G_{\bm \kappa}  \|_{\rm Spec}^2\left(\sum_{i=2}^k p_i\right)\\
&\leq \| G_{\bm \kappa}  \|_{\rm Spec}^2\, 2^{-m+1}.
\end{align}

These inequalities are easy to understand, once realized that $p_1\leq \frac{p_{\bm{\kappa}}}{2}$, to be maximally useful for sensing.  Their proof is given in the appendix \ref{Proof:equ:fisher_info_term_ineq}.
Note also the following inequality $\|G_{\bm \kappa}\|_{\rm Spec}\leq \|G\|_{\rm Spec}$.

Now we have the necessary tools to discuss the sufficient conditions for an exponential quantum advantage. We have seen that by entangling the sensors one can prepare a GHZ state inside the optimal DFS, yielding 
\begin{equation}
    \mathcal{F}_G^{\rm ent} = \max_{\bm \kappa} \|G_{\bm \kappa}\|_{\rm Spec}^2.
\end{equation}
In contrast for separable states  Eq.~\eqref{equ:FischerdecDFS} together with the bound~\eqref{equ:fisher_info_term_ineq} imply
\begin{equation}
\mathcal{F}_G^{\rm sep} \leq 2^{1-m} \sum_{\bm \kappa} \| G_{\bm \kappa}  \|_{\rm Spec}^2.
\end{equation}
Hence, an exponential advantage in $m$ always occurs if there is at most a polynomial number of decoherence-free subspaces with $\| G_{\bm \kappa}  \|_{\rm Spec}\neq 0$ (in particular, all trivial one-dimensional decoherence-free subspaces obviously give $\| G_{\bm \kappa}  \|_{\rm Spec}=0$). In full generality, we found that it is not straightforward to further simplify this statement in terms of the properties of the signal $G=G_0$ and noise $G_{j}$  generator.
However, one such case where the exponential advantage occurs is when the sensor network is of minimal size $n=m$. Here, there is only one two-dimensional DFS contributing to the QFI, while all the other ones are trivial, as shown in appendix~\ref{Proof:m=n_DFS}. 
One example of this exponential advantage was presented in~\cite{sekatskiOptimalDistributedSensing2020} another is experimentally demonstrated in~\cite{BateExperimental2024}. A similar advantage can be observed in the setting when the parameter is encoded via a correlated dephasing channel~\cite{wangExponentialEntanglementAdvantage2024}, which can also be analyzed with the bounds derived above.



\subsubsection{Numerical comparison  and entanglement of measurements}\label{sec:numeric}

In the previous section, we showed that an exponential advantage occurs when the number nontrivial DFS scales polynomially with the number of sensors. This is, however, not always the case, and moreover, it is not necessary for a quantum advantage. In particular, an aspect that we have so far ignored when analyzing the QFI of the final state is the entanglement of the measurements required to read it out. For the entangled GHZ strategies, we have seen that no entangled measurements are required. In contrast, product states are projected onto several DFS, and the optimal measurement needed to resolve different DFS (e.g. the POVM $\{\Pi_{\bm \kappa}\}$) and to saturate the QFI of a product state is in general entangled. 

We here explore these considerations with a specific example,  for which we compare the optimal entangled strategy with the one using the product state $\ket{+}^{\otimes n}$. 
The sensor configuration we consider is depicted in Fig.~\ref{fig:compare}(top). There are $n$ qubit sensors equally distributed on a circle\footnote{The orientation is such that the base of the resulting polygon is horizontal (parallel to the $x$-Axis).} with radius $10\pi$. The signal and noise fields 
\begin{equation}
f_j(\vec{x},t)=\beta_j \cos\left({\vec{k}_j^\intercal\Vec{x} -\omega t }\right) 
\end{equation}
are 2D plane-waves with the same frequency $\omega=1$ and different wave-vectors $\vec{k}_j=(\cos{\alpha_j},\sin{\alpha_j})$ parametrized by the angles $\alpha_j$. There are 19 noise field with the angles $\alpha_j$ equally distributed in the interval $\alpha_j\in[\frac{3\pi}{4}+\delta, \frac{7\pi}{4}+\delta]$, where $\delta=\frac{\pi}{20}$ is a small offset angle\footnote{Otherwise it can happen for some sensors size $n$ that the signal is linearly dependent with the noise waves on the sensor locations.}, and a single signal field with the angle $\alpha_0 = \frac{\pi}{4}$.

To find a good entangled strategy, a DFS is constructed by the method in Section~\ref{sec:DFS_control} (starting from the state $\bm z = (1,...,1)$). To optimize the QFI for the product state $\ket{+}^{\otimes n}$, one, in principle, needs to search through all possible control strategies, which is unfeasible. A simple naive solution uses the two-dimensional DFS derived for the entangled strategy. This is, however, unfair since the method is tailored to find a single DFS, while more DFSs might be realizable with other controls (which is favorable for product strategies, as we have seen in the last section). For a fair comparison, we thus slightly adjusted the method before applying it to product strategies. Concretely, we partition all the $n$ sensors into independent groups of size $m$ (minimal required to protect from noise). Then, for each of the groups, a DFS is constructed individually with the method of Section~\ref{sec:DFS_control}.

We investigate the sensing task of Fig.~\ref{fig:compare}(top) with $n$ sensors and a varying number $d$ of different noise sources, choosing $d$ such that the rank for the noise generators is $m-1$. The $d$ noise sources are a subset of the previously defined 19 directions.
Specifically for the first noise source $\alpha_1$, we choose the angle closest to $\alpha= \frac{\pi}{4}+\pi+\delta$, while the remaining $d-1$ noise are chosen with the angles closest to the signal $\alpha_0$. This construction ensures that the setup with $d+1$ contains one additional and all previous $d$ noise sources.  

The results of the comparison are reported in Fig.~\ref{fig:compare}(bottom). Here for varying number of noise sources parametrized by $m$ we change the number of sensors $n$ with three different scaling: minimal ($n=2 \ceil{\frac{m}{2}}$), linear ($n=4 \ceil{\frac{m}{2}}$) and quadratic ($n=4 \ceil{\frac{m}{2}}^2$).  Note that the scale of the QFI is logarithmic, and the QFI is divided by the number of sensors $n$. We see that in this example, the entangled strategy performs exponentially better than the classical approach in terms of the number of linearly independent noise generators (full lines vs dotted lines). In addition to computing the QFI of the product state, we also compute the classical Fisher information obtained with product states and product measurement by optimizing it with respect to all projective qubit measurements. Remarkably, the gap is even larger between the QFI of the product state and the classical Fisher information for the same states but obtained with product measurement (dotted lines vs points). This highlights that the QFI analysis of the product strategies is overly optimistic, as extracting it requires performing entangling measurements (which is not the case for the GHZ strategies).

\begin{figure}
    \centering
    \includegraphics[width=0.7\linewidth]{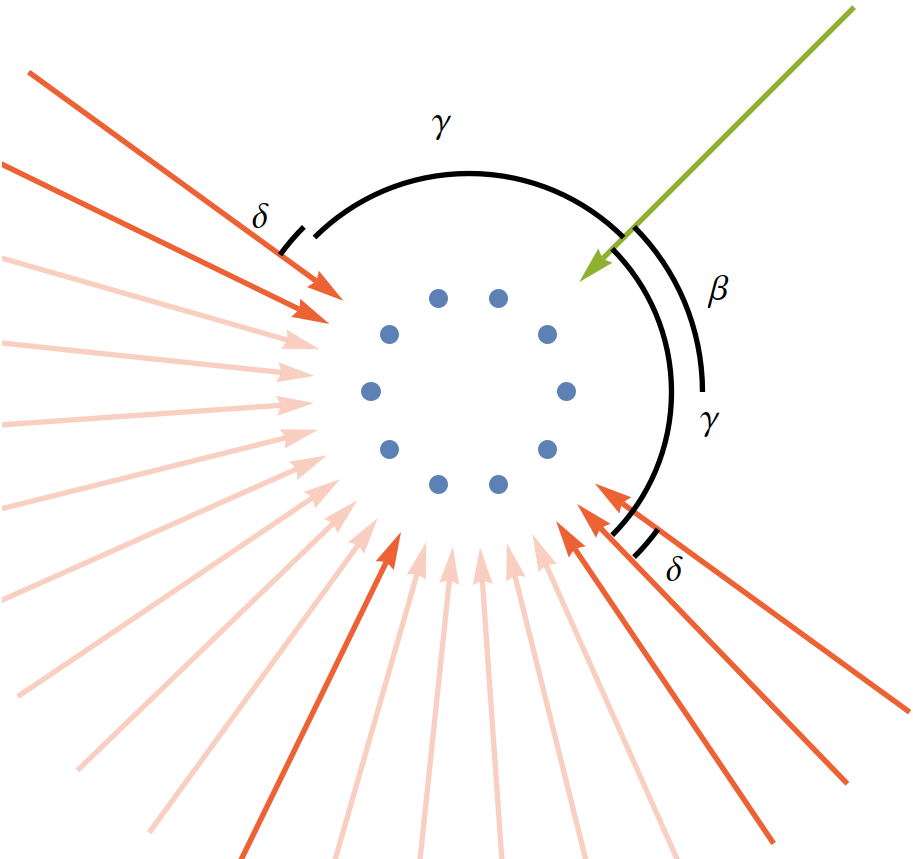}
    \includegraphics[width=\linewidth]{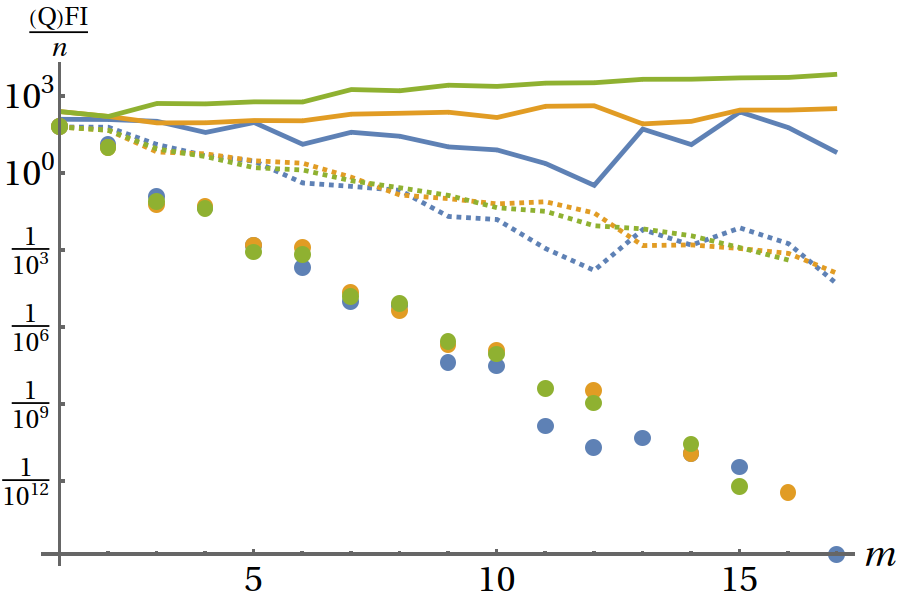}
        \caption{ \textbf{Top:} The setup for $n=10$ and $d=6$ noise waves. Sensors are placed on circle, with angle $2\pi/n$. \textbf{Bottom:} The Quantum Fisher information (QFI) and the Fisher information maximized over separable measurements (max F.I.) for an entangled and a product state for a different rank $m$ of the noise generators on the sensor locations. Each quantity was computed for constant, linear, and quadratically scaling sensor size and normalized to the number of sensor locations. The QFI equals the max F.I. for entangled states, and only the QFI is shown. The entangled state performs exponentially better than the separable state with entangling measurements (QFI) and separable measurements (max F.I.).
    }
    \label{fig:compare}
\end{figure}

\section{Examples}
\label{ref: examples}
In this section, we will present three examples to illustrate our methods. The first example considers a generic 6-sensor quantum network without additional symmetries. The second example considers a circular sensor with a point symmetry. There, we will show that the method can be used to construct an approximate DFS. The last example demonstrates the DFS construction through sensor placement.

\subsection{General DFS}\label{sec:example:generalDFS}

\begin{figure*}
	\begin{minipage}{0.45\linewidth}
		\centering
		\includegraphics[width=\linewidth]{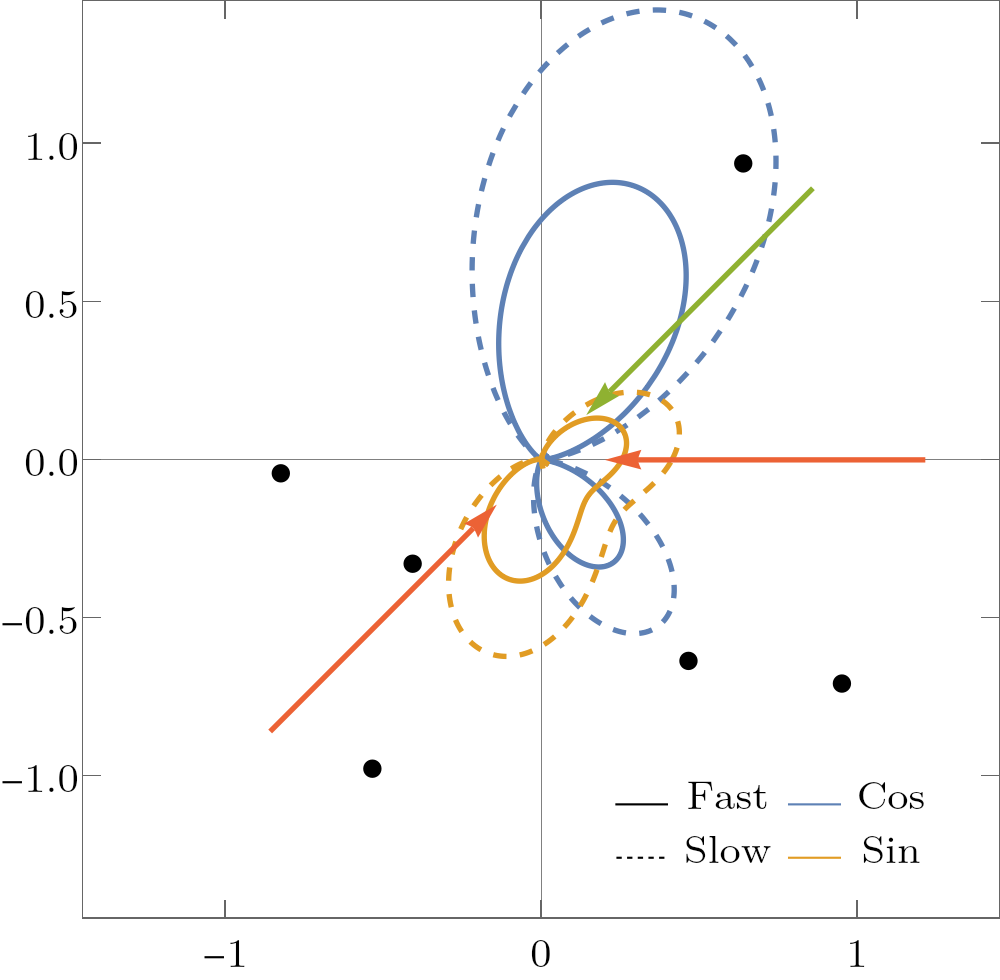}\\
		\includegraphics[width=\linewidth]{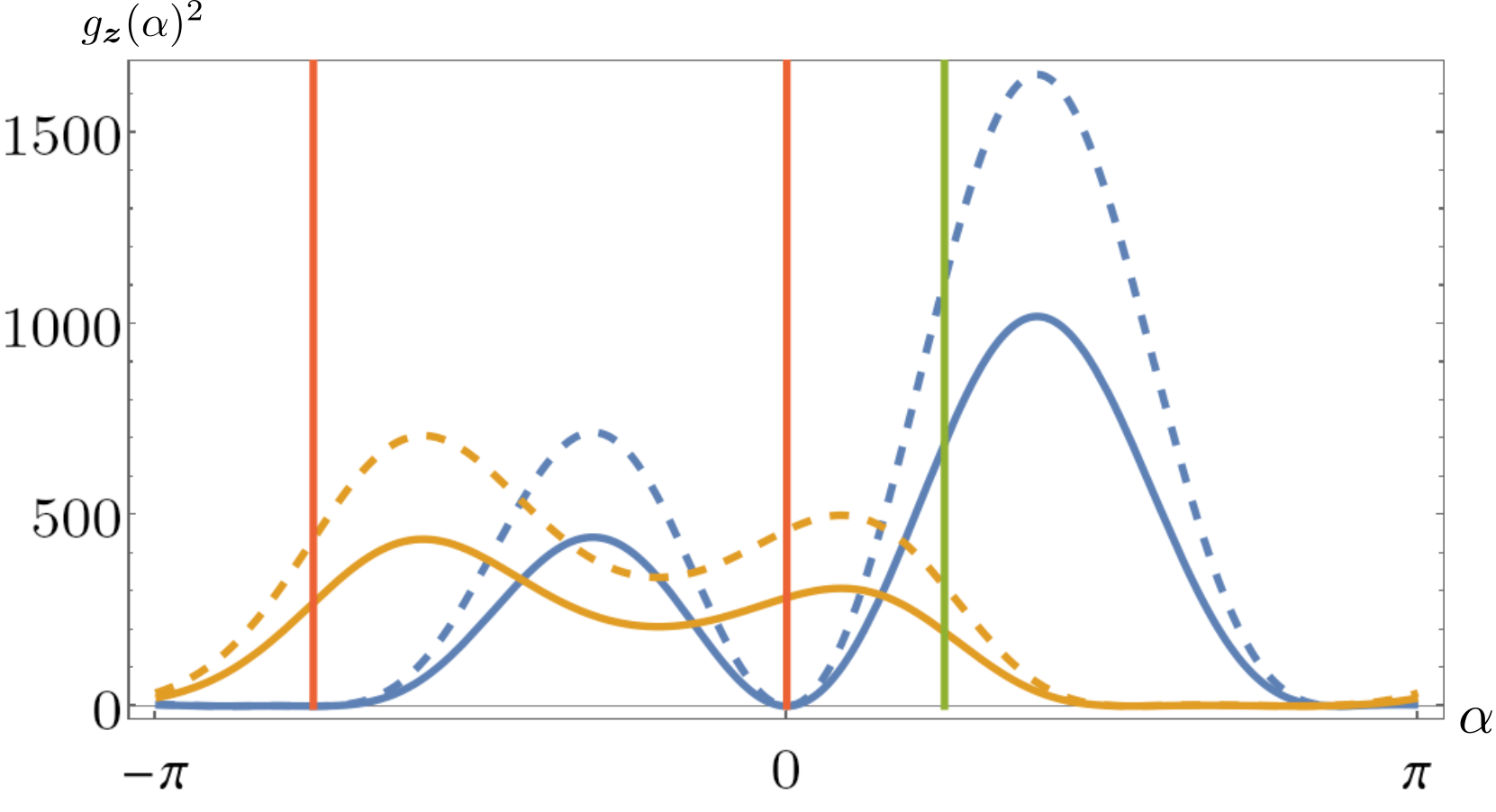}
		\caption{\textbf{Top}: Sensor positions and wave vectors for the general DFS example with known phases~\ref{sec:example:generalDFS}. 
        \textbf{Bottom}: The squared coupling $g_{\bm z}(\alpha)=g_{\bm z}(1,\sqrt{2}(\cos{\alpha},\sin{\alpha}),0)$ to cosine and sine waves as a function of the direction $\alpha$ for the fast and the slow control sequence for known phases. The green (orange) line marks the signal (noise) direction. The corresponding control sequences are shown in fig.~\ref{fig:example_general_phases_knwon_control}.
		}
		\label{fig:example_general_phases_knwon}
	\end{minipage}
	\hspace{1cm}
	\begin{minipage}{0.45\linewidth}
		\centering
		\includegraphics[width=\linewidth]{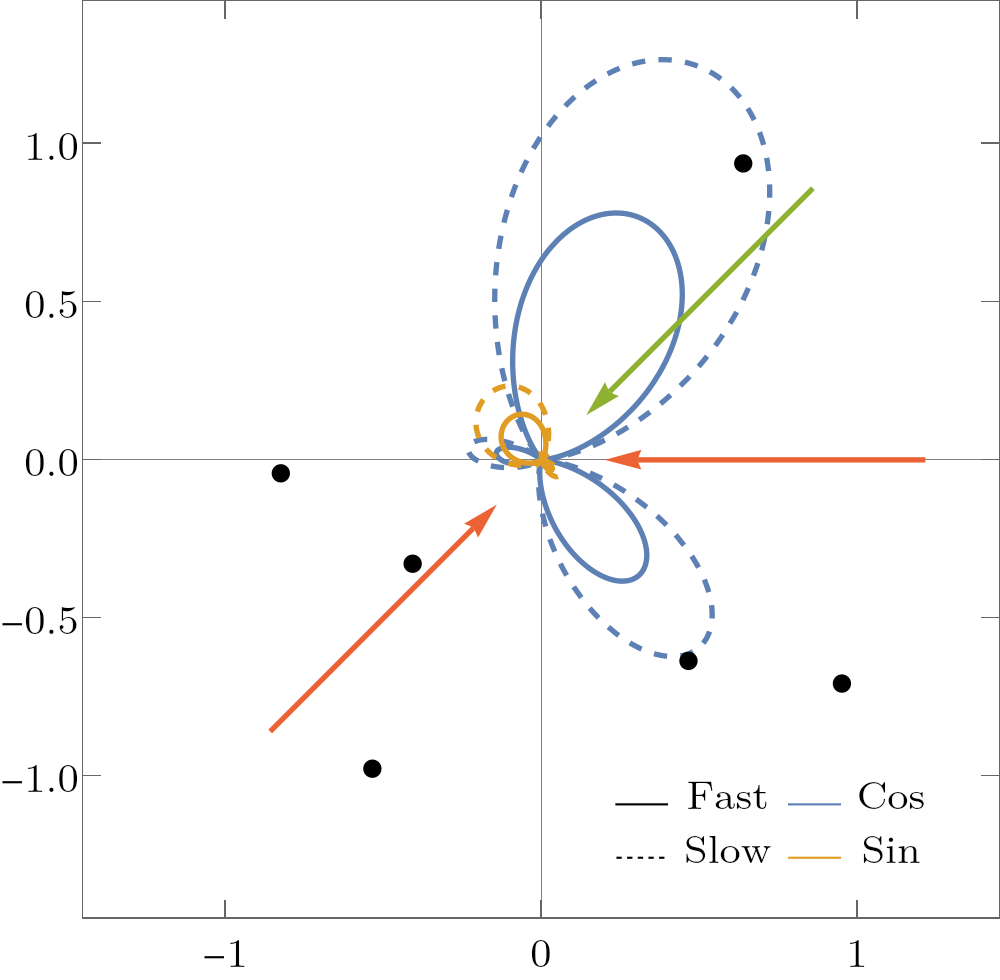}\\
		\includegraphics[width=\linewidth]{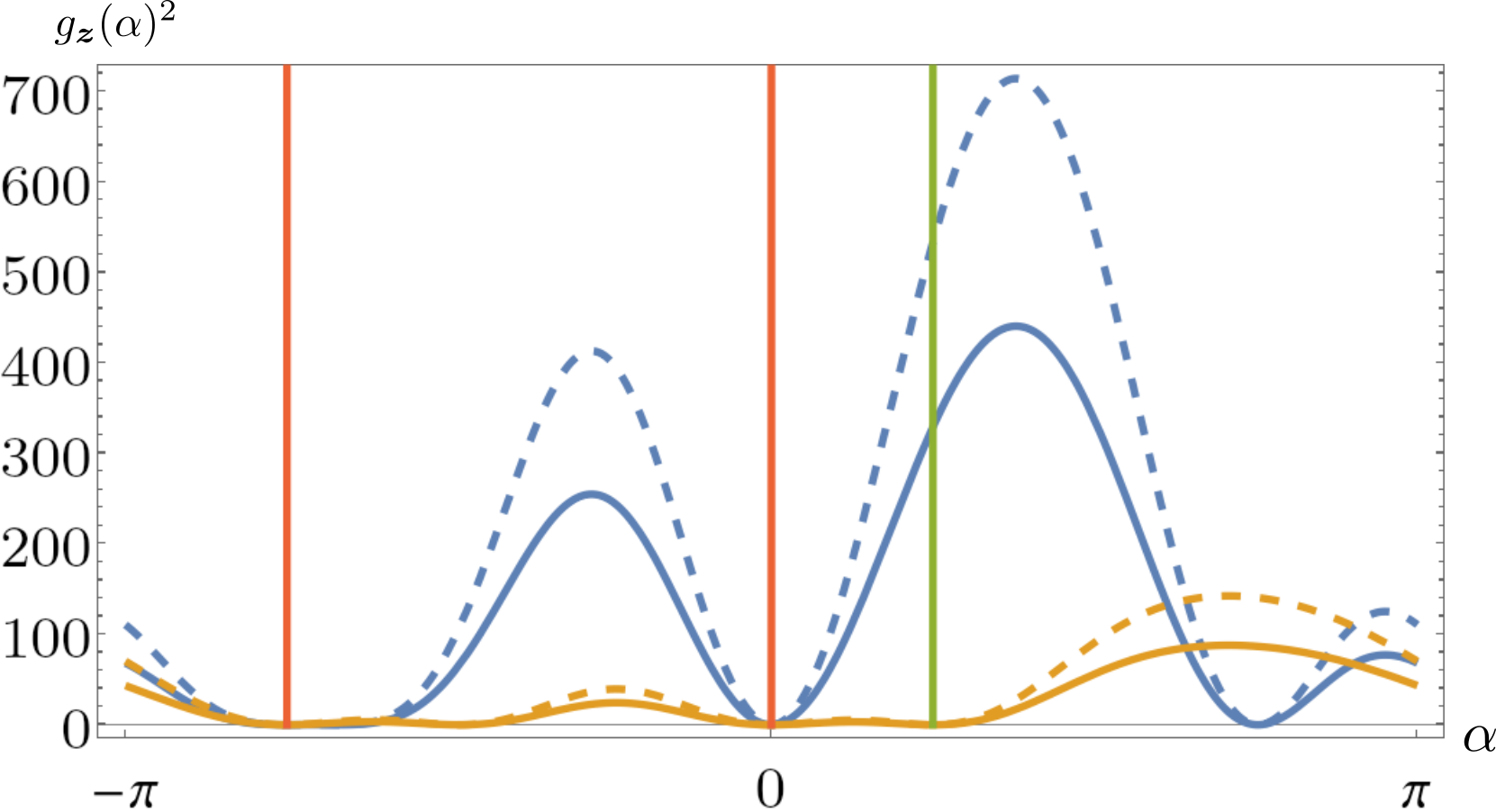}
		\caption{\textbf{Top}: Sensor positions and wave vectors of the general DFS example with unknown phases~\ref{sec:example:generalDFS}. 
        \textbf{Bottom}: The squared coupling $g_{\bm z}(\alpha)=g_{\bm z}(1,\sqrt{2}(\cos{\alpha},\sin{\alpha}),0)$ to cosine and sine waves as a function of the direction $\alpha$ for the fast and the slow control sequence for unknown phases. The green (orange) line marks the signal (noise) direction. The corresponding control sequences are shown in fig.~\ref{fig:example_general_phases_unknwon_control}.
		}
		\label{fig:example_general_phases_unknwon}
	\end{minipage}
\end{figure*}

We will use this example to demonstrate the three steps introduced in \ref{sec:DFS_control} and find a DFS for given sensor locations for a generic 6-sensor quantum network. The six arbitrarily chosen sensor locations are shown in fig.~\ref{fig:example_general_phases_knwon}(a). The frequency is the same for all generators and normalized to $\omega = 1$.
The interaction time $T=6\pi$ is three periods long. Three waves are interacting with the sensor network. One signal from the upper right $\vec{k}_s=\left(1,1\right)$ and two noise waves. One noise wave from the lower left $\vec{k}_1 = \left(-1,-1\right)$ is running in the opposite direction as the signal. The other comes from the right $\vec{k}_2 =\left(\sqrt{2},0\right)$. In Appendix~\ref{appendix:example:generalDFS}, we apply the three steps to construct a DFS from Sec~\ref{sec:DFS_control} explicitly.

In Fig.~\ref{fig:example_general_phases_knwon} we see the resulting squared coupling 
\begin{align}
g_{\bm z}(\omega=1,\vec{k}=\sqrt{2}(\cos{\alpha},\sin{\alpha}),\phi=0)^2
\end{align}
to a cosine ($\phi=0$) and sine ($\phi=\pi/2$) waves with propagation direction given by the angle $\alpha$, and  for the control sequences $C_k^{slow}$ and $C_k^{fast}$ for known phases. The control sequences themselves are shown in Fig.~\ref{fig:example_general_phases_knwon_control}. Fig.~\ref{fig:example_general_phases_unknwon} (Fig.~\ref{fig:example_general_phases_unknwon_control}) shows the coupling (control pulses) for $C_u^{slow}$ and $C_u^{fast}$ for unknown phases. 
As we expect, we see that the coupling to the cosine noise waves is zero in both cases, while for the sine noise wave in is only zero in the cases of unknown phases (in the first case such noise is ruled by the assumption on the phase). 
Additionally, the coupling for the slow sequence is as predicted $\frac{4}{\pi}$ higher than the fast sequence. The coupling sequence for the known phases cases additionally protects from cos wave coming from the left, i.e., $-k_2 =\left(-\sqrt{2},0\right)$ direction. In this example, protection from unknown phases comes with a cost of roughly a $\sqrt{2}$ factor in the coupling to this signal.

\subsection{Approximate DFS with a circular sensor network}  \label{sec:example:approx_dfs}
In some cases (e.g. if the signal filed is linearly dependent on the noises) all states in the DFS are necessarily decoupled from the signal. In these cases one can not be perfectly protected from noise and sensitive to the signal at the same time. Nevertheless, introducing and using the so-called approximate DFS 
\begin{equation}
	\mathrm{aDFS}_\epsilon = \left\{\bm z \Big \vert \vert g_{\bm z}(\omega_j, \vec{k}_j, \phi_j)  \vert < \epsilon \quad \forall 1\leq j \leq d \right\},
\end{equation}
can improve the sensor network significantly~\cite{hamannApproximateDecoherenceFree2022}. Analogously to the ideal situation, an approximate affine DFS can be defined.

In this example, we construct an approximate DFS using methods introduced in~\cite{hamannApproximateDecoherenceFree2022}. We use a point-symmetric circular sensor, where an even number of sensors is placed at angles $2\pi/n$ on a circle. A network with $n=16$ sensors is shown in Fig.~\ref{fig:circle-sensor}. We choose the states $\ket{1}^{\otimes n}$ and $\ket{-1}^{\otimes n}$ to span the DFS. As the sensor network is point symmetric, it is only sensitive to the cosine waves (see Sec.\ref{sec:point-sym}), and we can assume that all phases $\phi=0$ are zero. The signal is coming from the top right, and given by the wave vector $\vec{k}_s=\frac{1}{\sqrt{2}}(1,1)$. The signal and noise waves have the same frequency normalized to $\omega=1$. We now assume that the direction of the noise $\vec{k}_\alpha=(\cos(\alpha),\sin(\alpha))$ is not known exactly but it is restricted to lower left half-plane, i.e $\alpha\in \left[\frac{3\pi}{4},\frac{7\pi}{4}\right]$. To reduce the impact of such unknown noise, $n/2$ virtual noise waves are placed in the noisy half space, and the DFS is constructed to protect from these virtual noise waves exactly. The resulting sensitivity $g_{\bm z}(\omega=1,\vec{k}_\alpha,\phi=0)^2$ for different angles $\alpha$ is shown in fig.~\ref{fig:circle-sensor}(bottom). 

We see that the sensitivity for the virtual noise sources is zero. For arbitrary noise angles $\alpha$  the signal-to-noise ratio
\begin{equation} \label{eq: SNR}
{\rm SNR} = \frac{g_{\bm z}(\omega=1,\vec{k}_s,\phi=0)^2}{\underset{\frac{3}{4}\pi \leq \alpha\leq \frac{7}{4}\pi}{\max} g_{\bm z}(\omega=1,\vec{k}_\alpha,\phi=0)^2 }
\end{equation}
 is a relevant figure of merit to see how well the  aDFS can be used for sensing, as it is proportional to the optimal obtainable QFI~\cite{hamannApproximateDecoherenceFree2022}. The signal-to-noise ratio is depicted in Fig.~\ref{fig:circle-sensor}(c), one sees that it increases from $1$ for $n=2$ to $10^8$ for $n=26$. For completeness, we also show the remaining signal, which stays roughly constant around $10^3$ up to $n=12$ and then decays by three orders of magnitude to roughly $1$ for $n=16$. Notice that the loss in the signal can be compensated without changing the signal-to-noise ratio by increasing the interaction time or adding more qubits. 

\begin{figure}
	\centering
	\includegraphics[width=\linewidth]{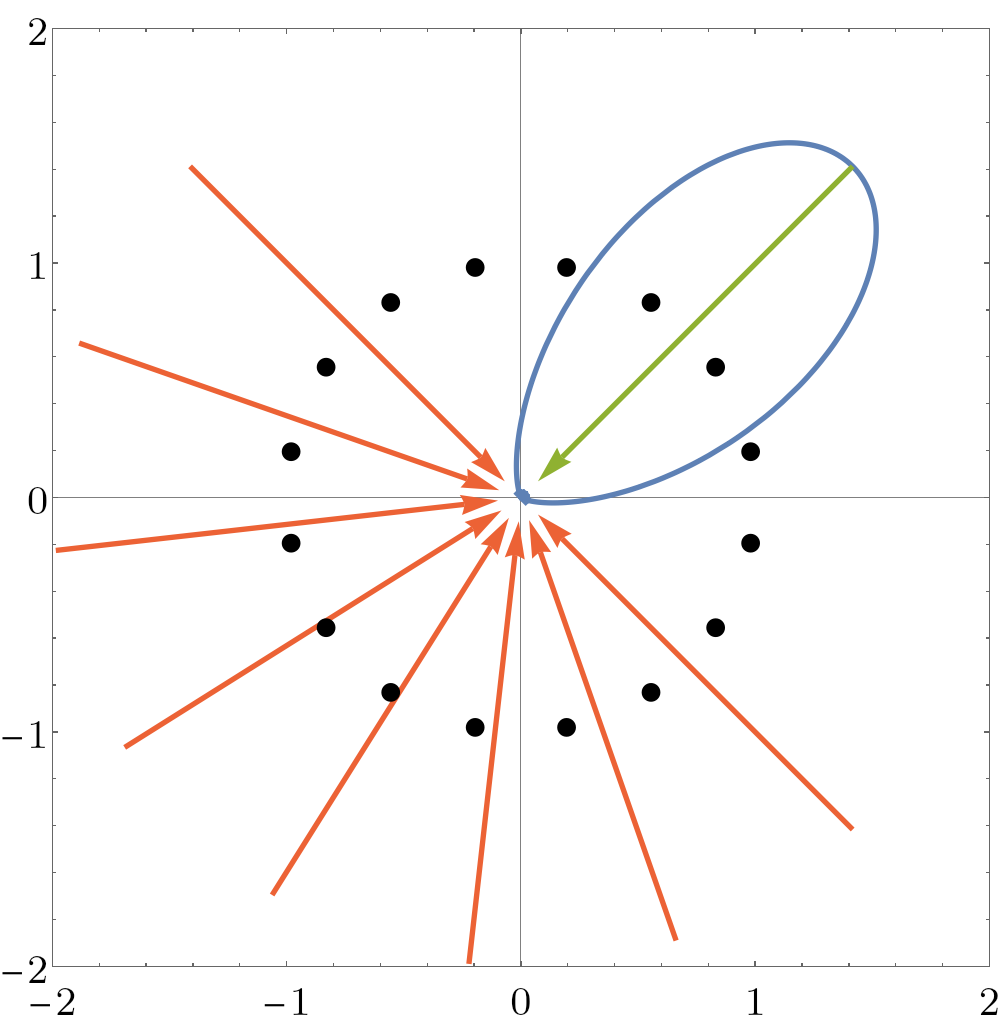}\\
	\includegraphics[width=\linewidth]{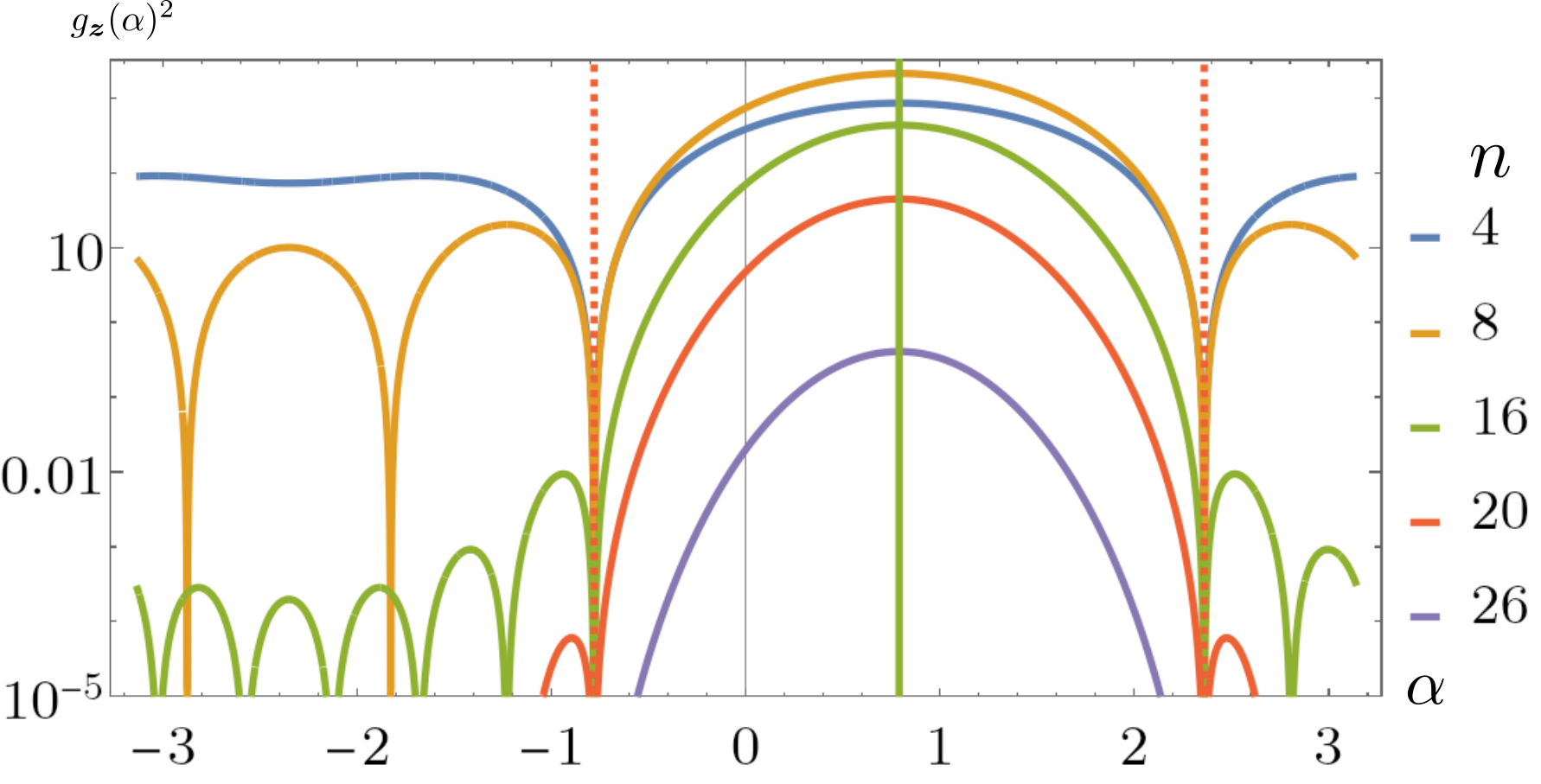}\\
	\includegraphics[width=\linewidth]{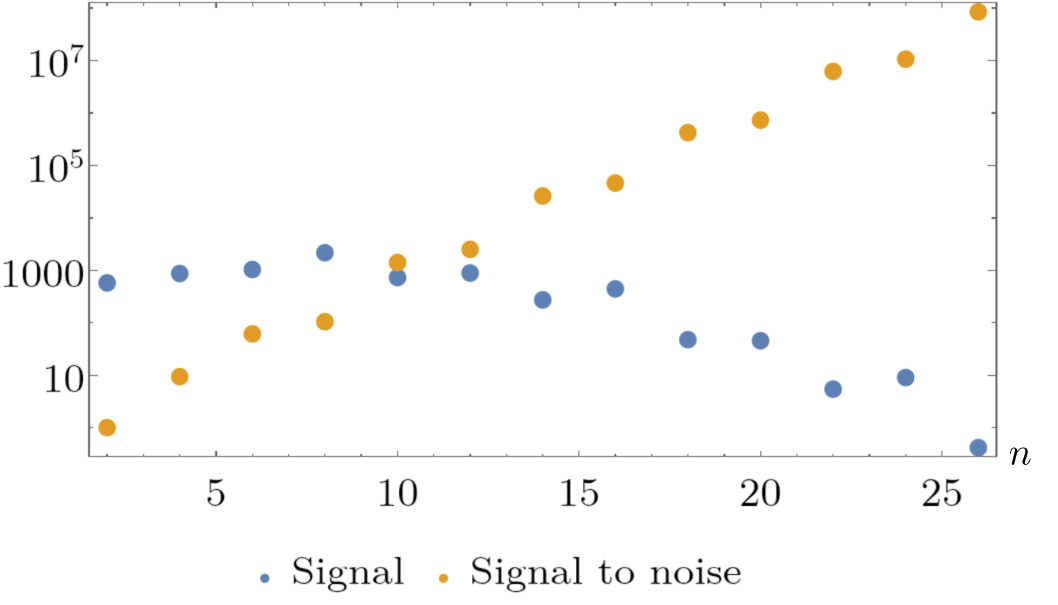}
	\caption{\textbf{Top:} The circular sensor (black point), signal (green arrow), and virtual noise (orange arrow) waves in the $n=16$ instance of the approximate DFS example (\ref{sec:example:approx_dfs}). The blue line shows the direction-dependent sensitivity.
 \textbf{Center:} Sensitivity as a function of the direction $\alpha$ for different $n$. The green (orange) lines mark the signal (noisy area). Virtual noises are perfectly suppressed to zero, while the remaining noisy area is highly suppressed. This suppression increases with $n$.  \textbf{Bottom:} Signal $g_{\bm z}(1,\vec{k}_s,0)^2$ and signal-to-noise ratio (Eq.~\eqref{eq: SNR}) as a function of $n$. The protection increases drastically with $n$, while the signal decays much slower.    }
	\label{fig:circle-sensor}
\end{figure}

\subsection{Sensor placement}
As discussed in  sec.~\ref{sec:no-control} it is possible to construct a DFS by just choosing the suitable sensor locations. In figure \ref{fig:example-sensor-placement}, we provide a example for three noise waves with a unit-length wave vector. The GHZ state has zero coupling to the noise sources and a high coupling/sensitivity in the remaining directions. 
\begin{figure}
	\centering
	\includegraphics[width=\linewidth]{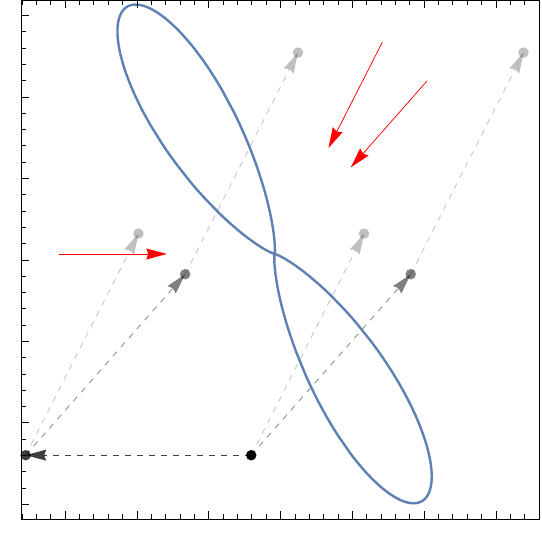}
	\caption{In this example, the sensor positions are chosen such that the sensor is insensitive to the three incoming red waves. The blue line shows the resulting direction-dependent sensitivity for a GHZ state. 
    }
	\label{fig:example-sensor-placement}
\end{figure}

\section{Generalizations}

\label{sec: generalization}

In this section, we discuss the application of our methods to time-dependent signals and noise beyond the case of plane waves.

\subsection{Other monochromatic waves}

The plane waves discussed so far are the most prominent example of field expansion in monochromatic waves
\begin{align}
    f_j(\vec{x},t,\phi) &= \Re( f_j(\vec{x})  e^{-i (\omega_j t + \phi_j)})
\end{align}
for which the space and time dependence factorize, and the latter has a fixed frequency. Plane waves are not the only relevant monochromatic wave expansion; e.g. spherical harmonics are another prominent one. As before if the phases $\phi_j$ of different fields are unknown, one can always write $f_j(\vec{x}) = | f_j(\vec{x}) | e^{\ii \varphi_j(\vec {x})}$,  split each field 
\begin{align*}
    f_j(\vec{x},t,\phi) &=\cos(\phi_j) f_j^{(1)}(\vec{x},t) + \sin(\phi_j) f_j^{(2)}(\vec{x},t) 
\end{align*}
into two monochromatic components $f_j^{(1)}(\vec{x},t)  = |f_j(\vec x)| \cos(\varphi_j(\vec{x})-\omega_j t)$  
with $  f_j^{(2)}(\vec{x},t)  = |f_j(\vec x)| \sin(\varphi_j(\vec{x})-\omega_j t)$, and treat those as independent  monochromatic fields.


It is not difficult to see that our methods from section~\ref{sec:DFS_control} naturally generalize to any monochromatic waves. 
One can define the field matrix
\begin{equation}
    F_{ji}=f_j(\Vec{x}_i,t,\phi_j),
\end{equation}
and construct the DFS by diagonalizing it. Since we are dealing with monochromatic waves, the time dependence of the generator is the same as for plane waves, hence one can employ the same control sequences, e.g. one can use the rectangular $\Pi$ function in Eq.~\eqref{eq: Pi cont}  again to construct the slow control sequence.



\subsection{Arbitrary time dependencies}
More generally one can consider the construction of DFS for fields $f_j(\vec{x},t)$ with an arbitrary temporal dependence. Notice that in this case one can not always  assume that the fields are synchronized. In general one must treat two delayed fields
\begin{equation}
f_j^{(\tau)}(\vec{x},t)=f_j(\vec{x},t+\tau),
\end{equation}
with different $\tau$ and $\tau'$, as coming from difference sources. Here, the delay parameter $\tau$ is analogous to the unknown phase $\phi$ in the case of waves, with the important difference that in general $f_j^{(\tau)}(\vec{x},t)$ can not be decomposed as a linear combination of two (or more) fields with fixed delays. Nevertheless, ideal noise-protected sensing is possible if the signal field $f_s(\vec{x},t)$ (restricted on sensor's positions) is not a linear combination of the noise fields $f_j^{(\tau)}(x,t)$ for all possible delays, for example, if the time translations of noise fields spans a finite-dimensional space. Even when this is not the case, one can construct an approximate DFS, by imposing protection from fields with finitely many different delays. This approach was shown to be very beneficial for sensing in the case where the position of the noise sources is unknown (spatially translated fields)~\cite{hamannApproximateDecoherenceFree2022}.\\

Let us now focus on the case of synchronized fields.
On the one hand, when fast control is available one can always construct the DFS by orthogonalizing the field matrix $F_{ji} = f_j(\vec{x}_i,t)$,
where the first $d$ rows correspond to the noise fields and the last row to the signal field. 
The fast control sequence is then given by the perpendicular signal component, e.g. the last row of the orthogonalized field matrix. On the other hand, given such a fast control sequence for arbitrary fields it is not straightforward to construct a slow control sequence. Designing a general and efficient method to do so remains an open problem.

\subsection{Higher dimensional DFS}
In certain cases,  most notably for Bayesian estimation~\cite{wolkNoisyDistributedSensing2020}, it is advantageous to construct DFS with a dimension higher than two. In this section, we present two methods to obtain these higher dimensional DFS. 

\subsubsection{Artificially adding noise generators to the field matrix}\label{sec:high_dim_dfs_noise}
Let $\bm{z}^l$ be a set of $L$ distinct strings labeling the states we want within the DFS, with $\bm z^1 = (1...1)$. Accordingly to Eq.~\eqref{eq: G int}, to construct such a DFS we need to find a control sequence $\bm C$ satisfying 
\begin{equation}\label{eq: hight DFS}
\int_{-\frac{T}{2}}^{\frac{T}{2}} \sum_{i=1}^n z_i^l f_j(\vec{x}_i,t) C_i(t) dt = 0
\end{equation}
for all noise fields $f_j$ and all states $\ket{\bm z^l}$. This can be done as follwos

First, for each state $\bm z^l$ we construct the field matrix
\begin{equation}
    F^{\bm{z}^l}_{ji} = z_i^l f_j(\vec{x}_i,t)
\end{equation}
of the $d$ noise fields. Then, we define a global field matrix $F$ by stacking all $F^{\bm{z}^l}$ contains the noise matrices for every state
$$F = \begin{pmatrix}
    F^{\bm z^1}\\
    \vdots\\
    F^{\bm z^l}\\
    F^{\mathrm{signal}}
\end{pmatrix}
$$
and adding the last row vector
\begin{equation}
    F^{\mathrm{signal}}_i = f_s(\vec{x}_i,t) 
\end{equation}
 corresponding to the signal field. The fast control sequence $\bm C^{fast}$ satisfying Eq.~\eqref{eq: hight DFS} can then be found by orthogonality $F$ with the scalar product
\begin{equation}
    \langle x,y\rangle = \sum_i \int_{-\frac{T}{2}}^{\frac{T}{2}} x_i(t)y_i(t) dt.
\end{equation}
The last row of the orthogonal matrix can be used to define the control sequence 
$\bm{C}^{fast}$ (after normalization). Note that this procedure also grantees that at least the state $\frac{1}{\sqrt 2}(\ket{\bm z^1}+\ket{-\bm z^1})$ inside the DFS retains some sensitivity to the signal (unless $F_i^{\rm signal}$ is not linearly independence from all $F_{ji}^{\bm z^l}$).


As last remarks, notice that given the fast control sequence one can again define a slow control one as discussed in \ref{eq: time control}. Analogously the results on point symmetric sensor networks from sec~\ref{sec:point-sym} apply here, provided that all $\bm z^l$ obey the point symmetry.

\subsubsection{By adding a control system}

Let us assume that in addition to the sensing qubit (on which $\sigma_z^{(i)}$ acts) each sensor in the network contains an auxiliary qubit. These qubits can then be prepared in an entangled state of the form 
\begin{equation}
\alpha \ket{0}_L+\beta \ket{1}_L= \alpha \ket{0}^{\otimes n} + \beta \ket{1}^{\otimes n}.
\end{equation}
Using an initially entangled state of sensing qubits and local operations, any two DFS constructions containing the state $\bm{z}$ and $\bm{z'}$ and realized with control sequences $\bm C$ and $\bm C'$, can now be realized in superposition controlled  by the state of the delocalize logical qubit $\ket{0}_L$ or $\ket{1}_L$. This is done by initially preparing the ``dressed'' state $\frac{1}{\sqrt{2}}(\alpha \ket{0}_L(\ket{\bm z }+\ket{-\bm z }) + \beta \ket{1}_L(\ket{\bm z '}+\ket{-\bm z '})$, and 
replacing the local gate sequences by gates coherently controlled on the state of the auxiliary qubit at each sensor. This ensures that the state remains protected from noise during the evolution, while the phase accumulated by the state $\ket{\bm z}$ and $\ket{\bm z'}$ can vary depending on the residual strength of the signal in their respective DFS.

In particular, one can always choose the same initial states $\ket{\bm z'}=\ket{\bm z}$ and set the control sequence $\bm C' = \lambda \bm C$ to be a ``slowdown'' of the original one $\bm C$, with $0<\lambda< 1$ (for slow control one can equivalently map $\Pi_{\sin^{-1}{A_i}}\mapsto\Pi_{\sin^{-1}{\lambda A_i}}$.) In this case the evolved final state acquires phases 
$U(T) \ket{0}_L\ket{\pm \bm z} = e^{\mp \ii \beta_0 \varphi }\ket{0}_L\ket{\pm \bm z}$ and $U(T) \ket{1}_L\ket{\pm \bm z} = e^{\mp \ii \beta_0 \lambda \varphi }\ket{1}_L\ket{\pm \bm z}$, where $\varphi$ is the phase resulting from the residual signal in the original DFS. This strategy thus allows one to obtain a four-dimensional DFS where the signal generator is given by the operator  $\hat G_{\rm DFS} \propto (\ketbra{\bm z} -\ketbra{-\bm z})\otimes (\ketbra{0}_L+\lambda \ketbra{1}_L)$. By replicating this construction with local auxiliary systems of dimension $D$, one can obtain a $2D$-dimensional DFS with the signal generator 
\begin{align}
\hat G_{\rm DFS} \propto \bigoplus_{\ell=0}^{D-1} \lambda_\ell (\ketbra{\bm z,\ell}-\ketbra{-\bm z,\ell})
\end{align}
for any $1=\lambda_0\geq \lambda_1\dots\geq \lambda_{D-1}\geq 0$

\section{Conclusion}
\label{sec:conslusion}

In this work, we have investigated the detection of waves with quantum sensor networks.
We considered the setting where qubit sensors distributed at different locations $\vec{x}_i$ are coupled to the field of amplitude $f(\vec{x},t)$ via the interaction   
 $H(\vec{x}_i,t) \propto f(\vec{x_i},t) \sigma_z^{(i)}$.
Utilizing entanglement between the sensors and local dynamical decoupling techniques, we have established a method that allows one to protect the sensor network from the influence of several known noise waves while retaining its sensitivity to the signal waves. The method is easy to use and relies on matrix orthogonalization.

It is well known that entangling the sensors allows the sensitivity to the signal field to improve from a linear to a quadratic scaling in the number of sensors (Heisenberg scaling). Remarkably, in a noisy scenario with the number of noise fields comparable to the number of sensors an exponential advantage of quantum sensors networks over classical ones is possible. We here identified sufficient conditions for such an exponential advantage to arise.

The setting and methods we have presented are not specific to a particular platform realizing the sensors, nor to a specific physical realization of the qubit-field interaction. Recently, such a table-top sensor network was experimentally realized with trapped ions to probe magnetic fields~\cite{BateExperimental2024}. In the future, with the advent of large-scale quantum networks, one can envision the realization of large-scale sensor networks to e.g. detect gravitational waves via local time dilation.




\section*{Acknowledgments}
A.H., P.A. and W.D. acknowledge support from the Austrian Science Fund (FWF). This research was funded in whole or in part by the Austrian Science Fund (FWF) 10.55776/P36009, 10.55776/P36010 and 10.55776/COE1.  P.S. acknowledges support from the Swiss National Science Foundation (project 192244). For open access purposes, the author has applied a CC BY public copyright license to any author accepted manuscript version arising from this submission. Finanziert von der Europäischen Union.

\bibliography{ref.bib}

\appendix
\onecolumngrid
\section*{Appendix}

\subsection{Coupling to a wave}
The previous section has shown that the total time evolution factors into the evolution with respect to the individual waves, but what is the relation between the waves parameter, the sensor positions and the applied time control of the network?
This section investigates the relation of an arbitrary individual generator $\hat{G}(\omega,\vec{k},\phi)$ corresponding to a wave described by $\omega, \vec{k}$ and $\phi$. It is diagonal in the product basis $\ket{\bm z}$, and denote the corespondent real eigenvalues as $g_{\bm{z}}(\omega,\vec{k},\phi)$ such that 
\begin{equation}
    g_{\bm{z}}(\omega,\vec{k},\phi) \ket{\bm z} = \hat{G}(\omega,\vec{k},\phi)\ket{\bm z}.
\end{equation}

Any sensor configuration, specified by the sensor positions $\vec{x}_1,...\vec{x}_n$ and the local control sequences $C_i(t)$, is 
Therefore the sensor positions ${\vec{x}_1,...\vec{x}_n}$ and time-control $C_i(t)$ are summarized as a combined spacetime coupling 
\begin{equation}
    C_{\bm{z}}(\vec{x},t) = \mathrm{rect}\left(\frac{t}{T}\right)\sum_i z_i C_i(t) \delta(\vec{x} - \vec{x}_i).
\end{equation}
Here $\delta(\vec{x})$ is the Dirac-delta distribution in the three spatial dimensions. The rectangular function $\mathrm{rect}(x)=(\frac{1}{2}+x)\Theta(\frac{1}{2}-x)$, where $\Theta(t)$ is the Heaviside-function, constrains the interaction time.

With these definitions, the relationship between wave coupling and the spacetime coupling  
\begin{equation}
    g_{\bm{z}}(\omega,\vec{k},\phi) = \Re\left(e^{i\phi} \mathfrak{F}\left( C_{\bm{z}}(\vec{x},t)\right)\right)
\end{equation}
simplifies to a complex component of the Fourier transform
\begin{equation}
	\mathfrak{F}\left(f(\vec x,t)\right) := \int d\vec{x} dt \ e^{i (\Vec{x}\vec{k} - \omega t)}f(\vec{x},t).
\end{equation}

The expression in terms of the Fourier transform has two major advantages. First, when evaluating $C_{\bm{z}}$, one can use known analytical solutions instead of solving bulky integrals, or, when evaluating numerically, use building functions for the Fourier transform, which are often much faster and numerically reliable. 

Secondly, the expression allows for some intuitive understanding. The phase $\phi$ will lead to some global \enquote{pulse} behavior.
The Fourier transform relates to position and momentum operators that fulfill an uncertainty principle.
Similarly, the Fourier transform imposes an uncertainty principle between the diameter of the spacetime coupling $C_{\bm{z}}(\vec{x},t)$ and the precision of the wave coupling $g_{\bm{z}}(\omega,\vec{k},\phi)$. This principle states that the more precise the targeting of a frequency or a direction, the more spread the network must be in time or space. 
The following section demonstrates this on a single quit example.


\section{QFI and QFIM}\label{appendix:qfi}
\subsection{QFI for a GHZ state}
This section contains a straightforward generalization of \cite{hamannApproximateDecoherenceFree2022}
The QFI for a signal with $(\omega,\vec{k},\phi)$ and the GHZ state $\frac{\ket{\boldsymbol{z}}+\ket{-\boldsymbol{z}}}{\sqrt{2}}$
is given by
\begin{equation}
    \mathcal{F}_{\omega, \vec{k}, \phi} = 4 g_{\boldsymbol{z}}(\omega,\vec{k},\phi)^2 \vert d \vert^2  = 4 \left(\sum_i z_i \cos(\phi + \gamma_i) \vert\mathfrak{F}(C_i(\vec{x},t))\vert\right)^2 \vert d\vert^2
\end{equation}
with $\gamma_i = -\ii \log(\frac{\mathfrak{F}(C_i(\vec{x},t))}{\vert\mathfrak{F}(C_i(\vec{x},t))\vert})$ being the phase of $\mathfrak{F}(C_i(\vec{x},t))$

and with
\begin{equation}
    d=\int p(\omega, \vec{k}, \phi, \beta)e^{-\ii \beta g_{\boldsymbol{z}}(\omega, \vec{k}, \phi}) d\omega d\vec{k}d\phi d \beta,
\end{equation}
where $p(\omega, \vec{k}, \phi, \beta)$ is the probability density which encodes the distribution of the parameters $\omega$, $\vec{k}$, $\phi$ and $\beta$, which define the noise.
\subsection{QFIM for the initial state within DFS}
This is a straightforward generalization of \cite{hamannMultiparameter2024}.
The DFS is given by $\{\bm z | g_{\boldsymbol{z}}(\omega, \vec{k}, \phi) \ \forall \omega, \vec{k}, \phi \}$, where $\omega$, $\vec{k}$ and $\phi$ are the parameters descibing the noise waves. 
The QFIM elements for a signal $w=(\omega,\vec{k},\phi)$ and $w'=(\omega',\vec{k}',\phi')$ and for a initial state $\psi = \sum_{\bm s} c_{\bm s} \ket{\bm s}$ are given by
\begin{equation}
\mathcal{F}_{w,w'} = 4 \sum_{\bm z} |c_{\bm z}|^2 g_{\bm z}(\omega,\vec{k},\phi) g_{\bm z}(\omega',\vec{k}',\phi') -  4 \left(\sum_{\bm z} |c_{\bm z}|^2 g_{\bm z}(\omega,\vec{k},\phi)\right)\left( \sum_{\bm z} |c_{\bm z}|^2 g_{\bm z}(\omega',\vec{k'},\phi')\right).
\end{equation}
Notice that $d=1$ if and only (upto $2\pi$ periodicity) if $g_{\bm z}=0 \ \forall \, p(\omega,  \vec{k},\phi,\beta)>0 \, \text{and} \,  \beta >0$. 
\subsection{QFI for classical strategies}
\begin{align}
	\mathcal{F}^+_{w} & := 4 \sum_{\boldsymbol{z}\in\{\pm1\}^N} \frac{1}{2^N} g_{\bm z}(\omega,\Vec{k},\phi)^2                                                                                                                                 \\
	                  & = 4 \sum_{\boldsymbol{z}\in\{\pm1\}^N} \frac{1}{2^N} g_{\bm z}(\omega,\Vec{k},\phi)^2                                                                                                                               \\
	                  & =  4 \sum_{\boldsymbol{z}\in\{\pm1\}^N} \frac{1}{2^N} \mathrm{Re}\left(\sum_{i=1}^N z_i e^{\ii \phi}\mathfrak{F}(C_i(\vec{x},t))\right)^2                                                                           \\
	\intertext{introducing $\bm F = \mathrm{Re}( e^{\ii \phi}(\mathfrak{F}(C_1(\vec{x},t)), \mathfrak{F}(C_2(\vec{x},t)),...)^\intercal)  $}
	                  & = \frac{4}{2^N} \sum_{\boldsymbol{z}\in\{\pm1\}^N} (\bm z^\intercal \bm F )^2                                                                                                                                       \\
	                  & =  \frac{4}{2^N} \sum_{\boldsymbol{z}\in\{\pm1\}^N}  (\bm F^\intercal \bm z )(\bm z^\intercal \bm F ) =  \frac{4}{2^N} \bm F^\intercal  \left(\sum_{\boldsymbol{z}\in\{\pm1\}^N} \bm z \bm z^\intercal\right) \bm F \\
	                  & =  4 \vert \bm F \vert^2                                                                                                                                                                                            \\
	                  & = \sum_i \cos^2(\phi + \gamma_i) \vert\mathfrak{F}(C_i(\vec{x},t))\vert^2                                                                                                                                           
	\intertext{with $\gamma_i = -\ii \log(\frac{\mathfrak{F}(C_i(\vec{x},t))}{\vert\mathfrak{F}(C_i(\vec{x},t))\vert})$ being the phase of $\mathfrak{F}(C_i(\vec{x},t))$}
    \intertext{For $C_i(\vec{x},t) = \delta(\vec{x}-\vec{x}_i \cos(\phi_i-\omega*t)$}
    &\propto \sum_i \cos(\phi - \phi_i +\vec{k}^\intercal\vec{x}_i)^2
\end{align}

%

\section{Overlap of the slow control sequence}\label{appendix:slow_control}

\begin{align*}
	\langle C^{fast}_i, f\rangle = \frac{1}{T}\int_{-\frac{T}{2}}^{\frac{T}{2}} C^{fast}_i(t) \cos(\underbrace{\vec{k}^\intercal \vec{x}_i +\phi}_{\alpha_i} - \omega t) dt & = \frac{A_i}{2}\cos(\alpha_i - \varphi_i) \\
	\langle C^{slow}_i, f\rangle =\frac{1}{T}\int_{-\frac{T}{2}}^{\frac{T}{2}} C^{slow}_i(t) \cos(\alpha_i-\omega t)                                                        & = \frac{1}{T}\int_{-\frac{T}{2}}^{\frac{T}{2}} \Pi_{\arcsin(A_i)}(\omega t + \varphi_i) \cos(\alpha_i-\omega t) \\
	                                                                                                                                                                        & = \frac{2}{T}\int_{-\frac{\arcsin(A_i)}{\omega}-\frac{\varphi_i}{\omega}}^{\frac{\arcsin( A_i)}{\omega}-\frac{\varphi_i}{\omega}}\cos(\alpha_i-\omega t) - \underbrace{\frac{1}{T}\int_{-\frac{T}{2}}^{\frac{T}{2}}\cos(\alpha_i - \omega t)}_{0} \\
	                                                                                                                                                                        & = \frac{2}{T\omega}\left(\sin(\alpha_i - \arcsin(A_i)-\varphi_i)- \sin(\alpha_i +  \arcsin(A_i)-\varphi_i)  \right)                                                                                                                               \\
	                                                                                                                                                                        & = \frac{2}{\pi}\cos({\alpha_i - \varphi_i})\sin(\arcsin{A_i})                                                                                                                                                                                     \\
	                                                                                                                                                                        & = \frac{2}{\pi} A_i \cos({\alpha_i - \varphi_i}) = \frac{4}{\pi} \langle C^{fast}_i, f\rangle                                                                                                                                                     
\end{align*}

\section{Examples}
\subsection{General DFS}\label{appendix:example:generalDFS}
Here, we explicitly provide the 3 steps to construct a DFS for the scenario in Section~\ref{sec:example:generalDFS}.
\subsubsection{Determine the field matrix}\label{sec:example_general_phases_knwon}
The field matrix contains the local field strength and its temporal correlation. Determining it differs depending on whether the phases are known or unknown.
\paragraph{Known phases:}
We will start by assuming that the phases are known and 0. This leads to the field matrix $F_k$ in (\ref{equ:Fk}).

\paragraph{Unknown phases:}
For unknown phases, we have to add a row where we replace every noise cos wave with a sin wave. This leads to a field matrix $F_u$ in (\ref{equ:Fu}).
Notice that only the noise waves are duplicated, as the signal phase is assumed to be known/chosen. 

\begin{align}
    \label{equ:Fk}
    F_k&=\left(
    \begin{array}{cccccc}
        \cos (t-0.87) & \cos (t-1.51) & \cos (t-0.74) & \cos (t-0.17) & \cos (t+1.58) & \cos (t+0.24) \\
        \cos (t+1.16) & \cos (t+0.75) & \cos (t+0.57) & \cos (t-0.66) & \cos (t-0.90) & \cos (t-1.35) \\
        \cos (t+0.87) & \cos (t+1.51) & \cos (t+0.74) & \cos (t+0.17) & \cos (t-1.58) & \cos (t-0.24) \\
    \end{array}
    \right) \\
    \label{equ:Fu}
    F_u&=\left(
    \begin{array}{cccccc}
        \cos (t-0.87)  & \cos (t-1.51)  & \cos (t-0.74)  & \cos (t-0.17) & \cos (t+1.58)  & \cos (t+0.24)  \\
        \cos (t+1.16)  & \cos (t+0.75)  & \cos (t+0.57)  & \cos (t-0.66) & \cos (t-0.90)  & \cos (t-1.35)  \\
        -\sin (t-0.87)  & -\sin (t-1.51)  & -\sin (t-0.74)  & -\sin (t-0.17) & -\sin (t+1.58) & -\sin (t+0.24) \\
        -\sin (t+1.16) & -\sin (t+0.75) & -\sin (t+0.57) & -\sin (t-0.66) & -\sin (t-0.90)  & -\sin (t-1.35)  \\
        \cos (t+0.87)  & \cos (t+1.51)  & \cos (t+0.74)  & \cos (t+0.17) & \cos (t-1.58)  & \cos (t-0.24)  \\
    \end{array}
    \right).
\end{align}

\subsubsection{Orthogonalize $F$}
We choose the states $\ket{1}^{\otimes n}$ and $\ket{-1}^{\otimes n}$ to be within the DFS.
Then, we apply Gram-Schmidt with the scalar product 
 \begin{equation}
\langle x,y \rangle = \sum_{i=1}^n \int_{-3\pi}^{3\pi} x_i(t)y_i(t) dt
\end{equation}
to the rows of $F_k$ and $F_u$.
This will give us the corresponding orthogonal signal components

\begin{align}
	s_\perp^k & = 
	\left(
    \begin{array}{c}
        0.07 \cos (t+0.32) \\
        0.15 \cos (t+2.33) \\
        0.06 \cos (t+1.40) \\
        0.15 \cos (t+1.12) \\
        0.14 \cos (t-2.39) \\
        0.18 \cos (t+0.59) \\
    \end{array}
    \right)
	\intertext{for the known and}
	s_\perp^u & = 
	\left(
	\begin{array}{c}
    	0.10 \cos (t-0.46) \\
    	0.18 \cos (t+2.62) \\
    	0.04 \cos (t-1.46) \\
    	0.03 \cos (t+0.94) \\
    	0.21 \cos (t-2.11) \\
    	0.13 \cos (t+0.21) \\
	\end{array}
	\right)
\end{align}
for the unknown phases.

\subsubsection{Identify the control sequence}
We normalize the orthogonal signal components to obtain the fast control sequences. Here, the normalization is done for the amplitudes, such that after the normalization, the biggest amplitude is one.
The resulting fast control sequences
\begin{align}
	C_k^{fast} & = \left(
    \begin{array}{c}
        0.36 \cos (t+0.32) \\
        0.82 \cos (t+2.33) \\
        0.33 \cos (t+1.40) \\
        0.81 \cos (t+1.12) \\
        0.74 \cos (t-2.39) \\
        1.00 \cos (t+0.59) \\
    \end{array}
    \right)\\
	C_u^{fast} & = \left(                   
	\begin{array}{c}
    	0.49 \cos (t-0.46) \\
    	0.84 \cos (t+2.62) \\
    	0.20 \cos (t-1.46) \\
    	0.16 \cos (t+0.94) \\
    	1.00 \cos (t-2.11) \\
    	0.64 \cos (t+0.21) \\
	\end{array}
	\right)
	\intertext{are local oscillations. We replace the oscillations with the rectangular $\Pi$ pulses for the slow control.}
	C_k^{slow} & =  \left(\begin{array}{cc} 
    	\Pi_{\arcsin(0.36)} (t+0.32)\\
    	\Pi_{\arcsin{0.82}} (t+2.33)\\
    	\Pi_{\arcsin(0.33)}(t+1.40)\\
    	\Pi_{\arcsin(0.81)}(t+1.12)\\
    	\Pi_{\arcsin(0.74)}(t-2.39)\\
    	\Pi_{\pi/2}(t+0.59)
	\end{array}
	\right) \\
	C_u^{slow} & =  \left(\begin{array}{cc} 
    	\Pi_{\arcsin(0.49)} (t-0.46)\\
    	\Pi_{\arcsin{0.84}} (t+2.64)\\
    	\Pi_{\arcsin{0.20}}(t-1.46)\\
    	\Pi_{\arcsin(0.16)}(t+0.94)\\
    	\Pi_{\pi/2}(t-2.11)\\
    	\Pi_{\arcsin(0.64)}(t+0.21)
	\end{array}
	\right).
\end{align}

Figure~\ref{fig:example_general_phases_knwon_control} and \ref{fig:example_general_phases_unknwon_control} shows the filliping sequences for the example presented in section~\ref{sec:example:generalDFS}. 
\begin{figure*}
	\begin{minipage}{0.45\linewidth}
		\centering
		\includegraphics[width=0.49\linewidth]{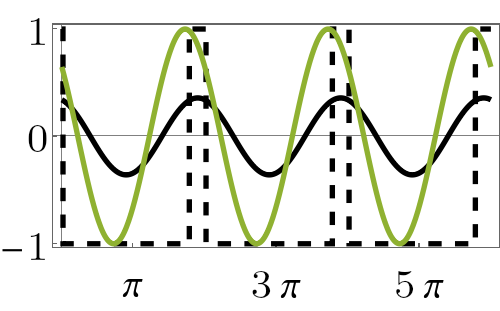}
		\includegraphics[width=0.49\linewidth]{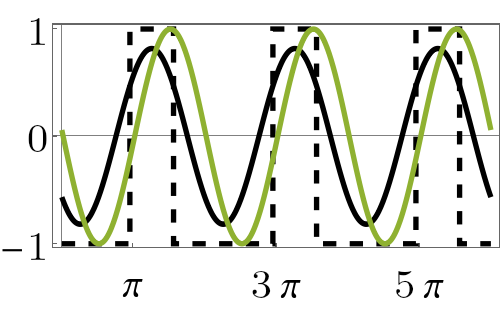}\\
		\includegraphics[width=0.49\linewidth]{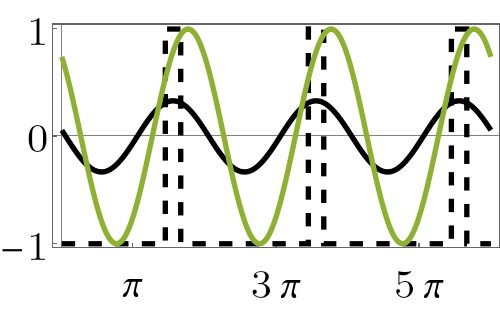}
		\includegraphics[width=0.49\linewidth]{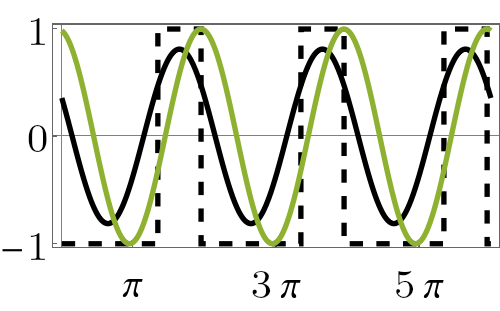}\\
		\includegraphics[width=0.49\linewidth]{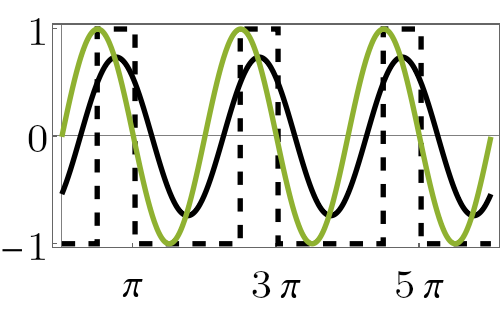}
		\includegraphics[width=0.49\linewidth]{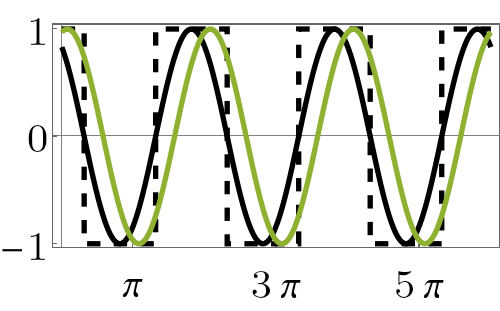}
		\caption{The plots show the control sequences $C_k^{fast}$,  $C_k^{slow}$ for the example \ref{sec:example:generalDFS} with known phases. The sensor locations are ordered from left to right and top to bottom. Additionally, the signal wave is shown. 
		}
		\label{fig:example_general_phases_knwon_control}
	\end{minipage}
	\hspace{1cm}
	\begin{minipage}{0.45\linewidth}
		\centering
		\includegraphics[width=0.49\linewidth]{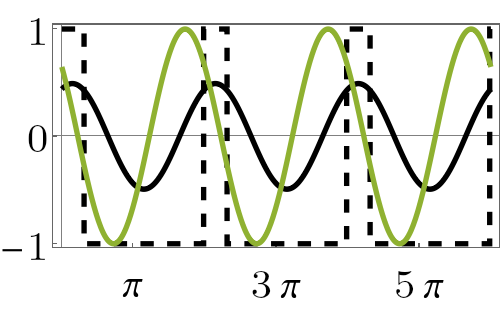}
		\includegraphics[width=0.49\linewidth]{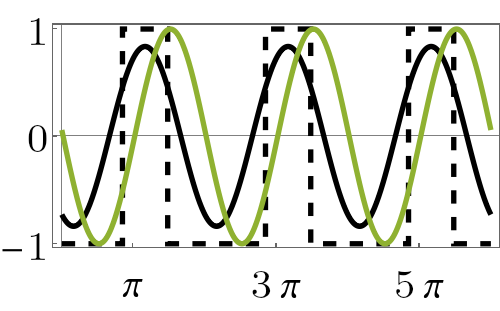}\\
		\includegraphics[width=0.49\linewidth]{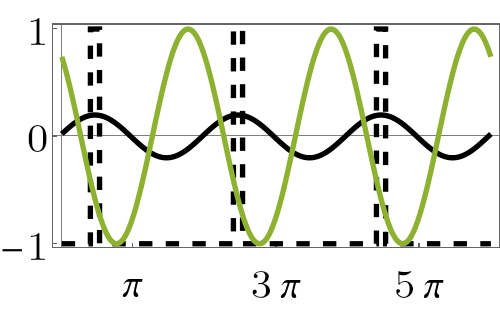}
		\includegraphics[width=0.49\linewidth]{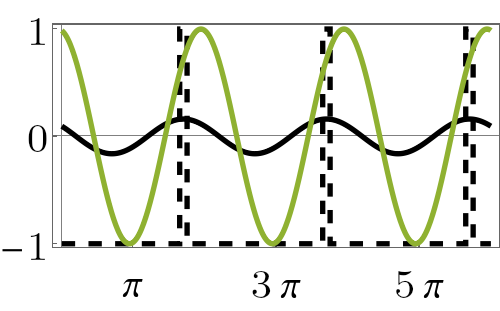}\\
		\includegraphics[width=0.49\linewidth]{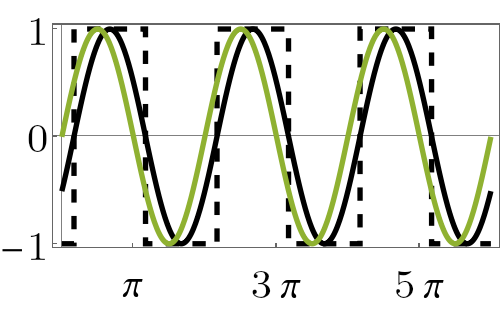}
		\includegraphics[width=0.49\linewidth]{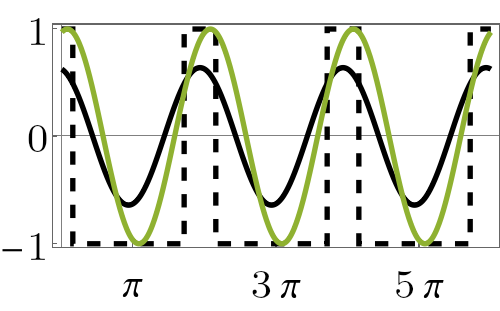}
		\caption{The plots show the control sequences $C_u^{fast}$,  $C_u^{slow}$ for the example \ref{sec:example:generalDFS} with unknown phases. The sensor locations are ordered from left to right and top to bottom. Additionally, the signal wave is shown. 
		}
		\label{fig:example_general_phases_unknwon_control}
	\end{minipage}
\end{figure*}

\section{Wave locking}\label{sec:wave_locking}
This method is motivated by the frequency locking. A flip is applied to each sensor whenever the signal sign (with $\phi=0$) changes.
Therefore flips are applied to the $i$-th sensor, at times
\begin{equation}
	\frac{\Vec{k}_s^\intercal\vec{x}_i}{\omega_s} +\left(\mathbb{Z}+\frac{1}{2}\right)\frac{\pi}{\omega_s},
\end{equation}
where $\omega_s$ and $\Vec{k}_s$ denote the signal frequency and the signal wave vector respectively.
A signal without phase gets effectively a positive function, and the phase will accumulate over the interaction time.
Other flips will mostly be off-resonance and average to zero for long interaction times. Therefore, it will effectively reduce the noise.

\section{Proofs and derivations}
\label{app: proof DFS}
\subsection{Preparatory results for separable strategies}

Let DFS be a set containing some bitstrings of lenght $n$. We say that it has Hamming distance $m$ 
\begin{equation}
    D_H(\bm z ,\bm z') \geq m \qquad \forall \bm z \neq \bm z'\in {\rm DFS}.
\end{equation}
We assume this property in the rest of the section.\\

To estimate the maximal number of elements in such set note that around each $\bm z \in {\rm DFS}$ we can define a Hamming distance Ball
\begin{equation}
B_d(\bm z) =\{ \bm x | D_H(\bm x,\bm z)\leq d\}.
\end{equation}
Furtheremore, for $d= \left\lfloor{\frac{m-1}{2}}\right\rfloor$
the balls $B_{\left\lfloor{\frac{m-1}{2}}\right\rfloor}(\bm z)$ and $B_{\left\lfloor{\frac{m-1}{2}}\right\rfloor}(\bm z')$ must be disjoint, otherwise there would exist a bitstring $\bm x$ such that
\begin{equation}
D_H(\bm z, \bm z')\leq D_H(\bm x, \bm z)+ D_H(\bm x, \bm z') \leq 2 \left\lfloor{\frac{m-1}{2}}\right\rfloor < m
\end{equation}
contradicting the assumption. Because the sets are disjoint, if there are $k$ elements in the DFS we must have
\begin{equation}
   k |B_{d(m)}(\bm z)| \leq 2^n \implies k \leq \frac{2^n}{|B_{d(m)}(\bm z)|} 
\end{equation}
For convenience let us also define the sets 
\begin{equation}
    C_d(\bm z):= \{\bm x | D_H(\bm x,\bm z)= d\},
\end{equation}
containing the bitstrings at a fixed Hamming distance $d$ form $\bm z$.
\bigskip

Now let us define a probability density $P({\bm x})$ on all bitstrings. We call it product if 
\begin{equation}
P({\bm x}) =\prod_{i=1}^n p_{x_i}^{(i)},\end{equation}
for some nonegative $p_{0(1)}^{(i)}$ with $p_{0}^{(i)}+p_1^{(i)}=1$. In the following we always have a product probability in mind.

Given a product probability density, take any $\bm z$ and $\bm z'$ from the DFS, by assumption they must be at Hamming distance at least $m$. Then $\Delta(\bm z,\bm z')$ be the set of position on which they differ, we have $|\Delta(\bm z,\bm z')|\geq m$. The strings can now be decomposed as $\bm z = \bm z_{|\Delta} \times \bm z_{|\Delta^C}$ and $\bm z' = \bm z'_{|\Delta} \times \bm z'_{|\Delta^C}$ with $ \bm z_{|\Delta}= \lnot \bm z'_{|\Delta}$ and $ \bm z_{|\Delta^C}= \bm z'_{|\Delta^C}$. For a product probability density, the probabilities of these strings factorize over the probailities of the substrings with 
\begin{align}\label{eq: whatever 1}
   P(\bm z) &= P(\bm z_{|\Delta^C}) P(\bm z_{|\Delta}) \\
   \label{eq: whatever 2}
    P(\bm z') &= P(\bm z'_{|\Delta^C}) P(\bm z'_{|\Delta}) =P(\bm z_{|\Delta^C}) P(\lnot \bm z_{|\Delta}).
\end{align}
To understand the relation between these probabilities we now show the following lemma

\begin{lemma} \label{Lemma:1}
For any product probability density over bitstrings of length $|\Delta|$, one has
\begin{equation}
P(\bm z)^\frac{1}{|\Delta|} + P(\lnot \bm z)^\frac{1}{|\Delta|}\leq 1.    
\end{equation}
\end{lemma}
\begin{proof}
For any fixed $P(\bm z)= \prod_{i=1}^{|\Delta|} p_i\in[0,1]$ let us maximize $P(\lnot \bm z)$. To do so we consider a deformation of $p_1,p_2$ leaving their product  and $P(\bm z)$ fixed. We have
\begin{align}
P(\lnot \bm z) =(1-p_1)(1-p_2) P(\lnot \bm z_{3}^{|\Delta|}) = -(p_1+p_2)P(\lnot \bm z_{3}^{|\Delta|}) +(1+p_1 p_2)P(\lnot \bm z_{3}^{|\Delta|}),
\end{align}
where the second term is not affected by the deformation.
We thus see that this probability is increased by decreasing $p_1+p_2$, which is minimized by setting them equal. By repeating argument of all pairs of indicies, we conclude that $P(\lnot \bm z)$ is maximized by setting $p_i = P(\bm z)^\frac{1}{|\Delta|}$, we thus find the
\begin{equation}
  P(\lnot \bm z) \leq  (1-P(\bm z)^\frac{1}{|\Delta|})^{|\Delta|} 
\end{equation}
to be tight for the iid probability distribution. With straightforward algebra one see that this bound is equivalent to 
\begin{equation}
P(\bm z)^\frac{1}{|\Delta|} + P(\lnot \bm z)^\frac{1}{|\Delta|}\leq 1,   
\end{equation}
concluding the proof.
\end{proof}

\begin{lemma} \label{Lemma:2} For a product probability density and two bitstrings $\bm z,\bm z'$ at Hamming distance $|\Delta|$ we have
    \begin{equation}\label{eq: corr pairwise}
        P(\bm z)^{\frac{1}{|\Delta|}}+
        P(\bm z')^{\frac{1}{|\Delta|}}\leq 1.
    \end{equation}
\end{lemma}
This follows from the lemma and the equations (\ref{eq: whatever 1},\ref{eq: whatever 2}), by imposing $P(\bm z|_{\Delta^C})\leq 1$.\bigskip

Now we would like to bound the probability of a Hamming distance Ball around a bitstrings. First we prove the following general bound.

\begin{lemma}  \label{Lemma:3}For any product probability density on bitstrings $P({\bm x})$ and any bitstring $\bm z$ one has 
\begin{align}
{\rm Pr}(\bm x \in B_m) &\geq F_{m;n}(P({\bm z})^{\frac{1}{n}})\\
F_{m;n}(p)&:= \sum_{k=0}^m \binom{n}{k} p^{n-k}(1-p)^k.
\end{align}
\end{lemma}

\begin{proof}
Without loss of generality we can set $\bm z$ to be the string of zeroes. Then we have
\begin{equation}
    P_{\bm z} = \prod_{i=1}^n p_i
\end{equation}
where $p_i = p_{0}^{(i)}$. 
The probability density is thus given by the vector $\bm p  = (p_{1}\dots p_n)$ of local probabilities. 
Keeping all of these values except $p_1,p_2$ fixed we consider the following deformation of the probability density
\begin{equation}
(p_1(\delta),p_2(\delta)) = (p_1 (1+\delta), p_2 \frac{1}{1+\delta})
\end{equation}
such that $P_{\bm z}$ and $p_1p_2$ remains unchanged. Let us analyze how this affects $\text{Pr}(\bm x \in B_m(\bm z))$ we have
\begin{align}
    &\text{Pr}(\bm x \in B_m(\bm z)) = p_1 p_2 \text{Pr}(\bm x_3^n \in B_m(\bm z_3^n))\\
    &(p_1 (1-p_2)+(1-p_1) p_2) \text{Pr}(\bm x_3^n \in B_{m-1}(\bm z_3^n))\\
    & +(1-p_1)(1-p_2)
     \text{Pr}(\bm x_3^n \in B_{m-2}(\bm z_3^n)) 
     \\
&= (p_1+p_2) (\text{Pr}(\bm x_3^n \in B_{m-1}(\bm z_3^n))-\text{Pr}(\bm x_3^n \in B_{m-2}(\bm z_3^n)))\\
&+\dots
\end{align}
where the dots contain all the terms which are not changed by the perturbation (independent of $p_1,p_2$ and proportional to $p_1p_2$). Since $(\text{Pr}(\bm x_3^n \in B_{m-1}(\bm z_3^n))-\text{Pr}(\bm x_3^n \in B_{m-2}(\bm z_3^n)))$ is positive we see that $\text{Pr}(\bm x \in B_m(\bm z))$ decreases by the perturbation iff $p_1+p_2$ decreases. This is the case whenever the perturbation makes $p_1$ and $p_2$ closer to each other . We conclude that $\text{Pr}(\bm x \in B_m(\bm z))$ decreased by setting  $p_1=p_2 =\sqrt{p_1 p_2}$.

This can be repeated for any pair of indices. Hence, for a given $P_{\bm z}$, the probability $\text{Pr}(\bm x \in B_m(\bm z))$ is minimized by the ``iid strategy'' $p_i =p=(P_{\bm z})^{\frac{1}{n}}$. In this case we find
\begin{align}
\text{Pr}(\bm x \in B_m(\bm z)) = \sum_{k=0}^m \binom{n}{k} p^{n-k} (1-p)^k
\end{align}
Subtracting $P_{\bm z}=p^n$ removes the term $k=0$ in the sum and concludes the proof.
\end{proof}

It is intuitive to see that this bound can be improved if we know that there is another bitstrings that has a higher probability that $\bm z$. This is fomralized by the following two results

\subsection{Lemmas required for Result~\ref{result:single_dfs_exponantially_small}}
\begin{lemma} \label{Lemma:4}For any product probability density on $m$-bitstrings (with $m\geq d \geq 0$), and some bitstrings $\bm z$ such that $P(\bm z)\geq P(\lnot \bm z)$ one has
\begin{equation}
{\rm Pr}(C_d(\lnot \bm z))\geq  P(\lnot \bm z)\binom{m}{d} \left(\frac{P(\bm z)}{P(\lnot \bm z)}\right)^\frac{d}{m} \geq P(\lnot \bm z)\binom{m}{d}.
\end{equation}
\end{lemma}
\begin{proof}
    Without loss of generality we can chose $\bm z=\bm 0$ so that $\lnot \bm z=\bm 1$. The product distribution is defined by $m$ probabilities $p_i := p_0^{(i)}$ with
    \begin{equation}
        P(\bm 0) = \prod_{i=1}^m p_i \qquad P(\bm 1) = \prod_{i=1}^m (1-p_i).
    \end{equation}

\textbf{Case} $d=0$ or $d=m$: Here the lemma is true by the assumption $P(\bm z)\geq P(\lnot \bm z)$.

\textbf{Case} $m=2$ and $d=1$: 
\begin{align}
    {\rm Pr}(C_1(\bm 1)) &\geq  2 P(\bm 1) \sqrt{\frac{P(\bm 0)}{P(\bm 1)}} \\
    \Longleftrightarrow 0&\leq 
    \left({\rm Pr}(C_1(\bm 1))\right)^2 - \left(2 P(\bm 1) \sqrt{\frac{P(\bm 0)}{P(\bm 1)}} \right)^2 \\
    &= \left( p_1 + p_2 -2 p_1 p_2\right)^2 - 4 p_1 p_4 (1-p_1)(1-p2)\\
    &= p_1^2+2 p_1 p_2 +p_2^2 \geq 0
\end{align}
\textbf{Case} $m \geq 3$:

Now let us also define $f_i = \frac{p_i}{1-p_i}$. By assumption we have 
\begin{equation}\label{eq: cons f}
\prod_i f_i = \frac{P(\bm 0)}{P(\bm 1)}\geq 1.
\end{equation}
Next, note that $C_d(\bm 1)$ contains all the bitstrings with $d$ zeroes and $m-d$ ones. To express the probability of $C_d(\bm 1)$ let us define the set $\mathds{S}^m$ containing all sets $S\subset \{1\dots,m\}$ with $d$ elements (with $S^C =\{1,\dots,m\}\setminus S$). There are $|\mathds{S}_d|=\binom{m}{d}$ such sets, and now we can write
\begin{equation}
{\rm Pr}(C_d(\bm 1))= \sum_{S\in \mathds{S}_d}\prod_{i\in S} p_i \prod_{i\in S^C} (1-p_i) = P(\bm 1) \sum_{S\in \mathds{S}_d}\prod_{i\in S} f_i.
\end{equation}
We want to lower bound this probability formally this can be done by defining the following program 
\begin{align}
    {\rm Pr}(C_d(\bm 1)) \geq  (\ast):=\min_{0\leq p_1,\dots, p_m \leq 1} \quad & P(\bm 1) \sum_{S\in \mathds{S}_d}\prod_{i\in S} f_i \\
    \text{such that} \quad & \prod_{i=1}^m p_i =P(\bm 0) \\ &
   \prod_{i=1}^m (1-p_i) = P(\bm 1)
   \\ &\prod_{i=1}^m f_i = \frac{P(\bm 0)}{P(\bm 1)}.
\end{align}
Here, the constraints are redundant since the first two imply the last one. However, the reason for writing it like this is that we have now relaxed the minimization problem by dropping the first two constraints. The minimum of the relaxed problem is necessarily lower than the original one; hence we get
\begin{align}
    (\ast)\geq  \min_{0\leq f_1,\dots, f_m} \quad & P(\bm 1) \sum_{S\in \mathds{S}_d}\prod_{i\in S} f_i \\
    \text{such that} \quad & \prod_{i=1}^m f_i = \frac{P(\bm 0)}{P(\bm 1)}. \label{eq: const app}
\end{align}
To solve this minimization, focus on two values, say $f_1$ and $f_2$. The probability can be written as 
\begin{align}\label{eq: P app}
\text P=P(\bm 1) \sum_{S\in \mathds{S}_d}\prod_{i\in S} f_i 
= P(\bm 1)(  f_1 f_2 \sum_{S\in \mathds{S}_{d-2} '}\prod_{i\in S} f_i +(f_1 + f_2) \sum_{S\in \mathds{S}_{d-1}''}\prod_{i\in S} f_i +\dots)
\end{align}
where $\mathds{S}_{d-2}'$ contains all subsets of $\{3,\dots,m\}$ with $d-2$ elements, $\mathds{S}_{d-1}''$  contains  all subsets of $\{3,\dots,m\}$ with $d-1$ elements, and $\dots$ collects all terms where neither $f_1$ nor $f_2$ appear. We consider the following deformation of these values
\begin{equation}
    (f_1,f_2) \to (f_1 (1+\delta) , \frac{f_2}{ (1+\delta)} ),
\end{equation}
which leaves the product $f_1,f_2$ unchanged an thus also does not break the constraint in Eq.~\eqref{eq: const app}.
Yet, it does change the value of P in Eq.~\eqref{eq: P app} to
\begin{equation}
    P \to P + \delta (f_1-f_2) \sum_{S\in \mathds{S}_{d-1}''}\prod_{i\in S} f_i + O(\delta^2).
\end{equation}
We see that P can be decreased by a positive $\delta$ (increasing $f_1$) if $f_2>f_1$. In other words, the value of P is decreased if we move the values of $f_1,f_2$ closer to each other. 
This argument can be repeated with all pairs of indices, leading to the conclusion that P is minimized when all $f_i$ are equal. By~\eqref{eq: const app} their values are then  
$f_i = \left(\frac{P(\bm 0)}{P(\bm 1)}\right)^\frac{1}{m}$.
Plugging these values back in the expression we finally find
\begin{equation}
{\rm Pr}(C_d( \bm 1))\geq  P(\bm 1) \sum_{S\in\mathds{S}_d}f_i^d =  P(\bm 1)\sum_{S\in \mathds{S}_d}\left(\frac{P(\bm 0)}{P(\bm 1)}\right)^\frac{d}{m} = P(\bm 1) \binom{m}{d} \left(\frac{P(\bm 0)}{P(\bm 1)}\right)^\frac{d}{m} \geq P(\bm 1) \binom{m}{d} .
\end{equation}
Which proves the lemma.
\end{proof}

We now use the lemma to derive the following corrollary.
\begin{lemma} \label{lemma: prob ball tight}
   Consider a product probability distribution and two n-bitstrings $\bm z$ and $\bm z'$, such that $P(\bm z)>P(\bm z')$ and $D_H(\bm z,\bm z')=|\Delta|$ (with $|\Delta|>d\geq 0$). For any $d < |\Delta|$ we have
   \begin{equation}
       {\rm Pr} (B_d(\bm z')) \geq P(\bm z') \sum_{k=0}^d \binom{|\Delta|}{k}
   \end{equation}
\end{lemma}
\begin{proof}
Let $\Delta$ be the set of indices on which the strings differ, they can then be  factored as $\bm z = \bm z_{|\Delta} \times \bm z_{|\Delta^C}$ and $\bm z' = \bm z'_{|\Delta} \times \bm z'_{|\Delta^C}$, with $\bm z'_{|\Delta} =\lnot \bm z_{|\Delta}$ such that 
\begin{align}
   P(\bm z) &= P(\bm z_{|\Delta^C}) P(\bm z_{|\Delta}) \\
    P(\bm z') &= P(\bm z'_{|\Delta^C}) P(\bm z'_{|\Delta}) =P(\bm z_{|\Delta^C}) P(\lnot \bm z_{|\Delta}).
\end{align}
and $P(\lnot \bm z_{|\Delta^C}) \leq P (\bm z_{|\Delta^C})$ by assumption. The ball $B_d(\bm z'_{|\Delta})$ contains all the bitstrings where only up to $d$ the bits within $\Delta$ have been flipped. Hence we find
\begin{align}
    {\rm Pr} (B_d(\bm z')) &\geq  P(\bm z'_{|\Delta^C}) {\rm Pr} (B_d(\bm z'_{|\Delta})) \\
    &\geq  P(\bm z'_{|\Delta^C}) \sum_{k=0}^d {\rm Pr}(C_d(\bm z'_{|\Delta}))\\
&\geq P(\bm z'_{|\Delta^C})\sum_{k=0}^d P(\bm z'_{|\Delta}) \binom{|\Delta|}{k} \left(\frac{P(\bm z_{|\Delta})}{P(\bm z'_{|\Delta})}\right)^{\frac{k}{\Delta}} \\
&= \sum_{k=0}^d P(\bm z'_{|\Delta^C})P(\bm z'_{|\Delta}) \binom{|\Delta|}{k} \left(\frac{P(\bm z_{|\Delta})}{P(\bm z'_{|\Delta})}\right)^{\frac{k}{\Delta}} \left(\frac{P(\bm z_{|\Delta^C})}{P(\bm z'_{|\Delta^C})} \right)^{\frac{k}{\Delta}}\\
& = \sum_{k=0}^d \binom{|\Delta|}{k} P(\bm z)^{\frac{k}{|\Delta|}}P(\bm z')^{1-\frac{k}{|\Delta|}}\\
&\geq P(\bm z') \sum_{k=0}^d \binom{|\Delta|}{k} .
\end{align}
Where in in the line before the last we give a bound which is more cumersome but tighter than in the statement of the lemma.
\end{proof}

\subsection{Proof of the Result~\ref{result:single_dfs_exponantially_small}}\label{section:single_dfs_exponantially_small}
\begin{proof}\label{proof:single_dfs_exponantially_small}
    Let us now come back to our set DFS= $\{\bm z_1,\dots,\bm z_k\}$. For a given product probability density, we would like to understand what are the possibility values of $P(\bm z_i)$. Without loss of generality we can arrange the vecotors such that these probabilities are decreasing $P(\bm z_i)\geq P(\bm z_{i+1})$. 
    
    We have seen that the balls $B_{\left\lfloor{\frac{m-1}{2}}\right\rfloor}(\bm z_i)$ around the bitsrings $\bm z_i$ are disjoint. Moreover each $z_i$ with $i\geq 1$ is at Hamming distance $|\Delta|\geq m$ away from $\bm z_1$, and $P(z_i)\leq P(z_1)$. Hence by lemma~\ref{lemma: prob ball tight} it follows that for all $i\geq 2$
    \begin{align}
        {\rm Pr}( B_{\left\lfloor{\frac{m-1}{2}}\right\rfloor}(\bm z_i)) \geq P(\bm z_i) \sum_{l=0}^{\left\lfloor{\frac{m-1}{2}}\right\rfloor} \binom{|\Delta|}{l} \geq P(\bm z_i) \sum_{l=0}^{\left\lfloor{\frac{m-1}{2}}\right\rfloor} \binom{m}{l}.
    \end{align}
    And since these balls are disjoint we must also have
    \begin{align}
    1\geq \sum_{i=1}^k {\rm Pr}( B_{\left\lfloor{\frac{m-1}{2}}\right\rfloor}(\bm z_i)) \geq  \sum_{i=2}^k {\rm Pr}( B_{\left\lfloor{\frac{m-1}{2}}\right\rfloor}(\bm z_i)) \geq \left(\sum_{i=2}^k P(z_i)\right)\left(\sum_{l=0}^{\left\lfloor{\frac{m-1}{2}}\right\rfloor} \binom{m}{l}\right).
    \end{align}
    Or simply 
    \begin{equation}
      \sum_{i=2}^k P(z_i) \leq  \frac{1}{\sum_{l=0}^{\left\lfloor{\frac{m-1}{2}}\right\rfloor} \binom{m}{l}}
    \end{equation}
    
    This essentially proves the result~\ref{result:single_dfs_exponantially_small} in the main text.
\end{proof}

\subsection{Further derivations within Section~\ref{sec:product_general}}
\subsubsection{Explicit derivation of equation (\ref{equ:FischerdecDFS})} \label{Proof:equ:FischerdecDFS}
\begin{align*}
 \mathcal{F}_G &= \mathcal{F}\left(  \bigoplus_{{\rm DFS}_{\bm \kappa}} \Pi_{\bm \kappa} \ket{\Psi}_0 \bra{\Psi}_0 \Pi_{\bm \kappa}\right) \\
 &= \sum_{\bm  \kappa} p_{\bm \kappa} \mathcal{F}(\ket{\Psi}_{\bm \kappa} \bra{\Psi}_{\bm \kappa}) \ \text{(QFI linear in direct sums)} \\
 &= 4\sum_{\bm  \kappa} p_{\bm \kappa} {\rm Var}_{\Psi_{\bm  \kappa}} (G)  \ (\text{QFI formula}) \\
 &= 4\sum_{\bm  \kappa} p_{\bm \kappa} \left(\frac{\bra{\Psi}_0 \Pi_{\bm \kappa} G^2 \Pi_{\bm \kappa} \ket{\Psi}_0  }{p_{\bm \kappa}}- \frac{\bra{\Psi}_0 \Pi_{\bm \kappa} G \Pi_{\bm \kappa} \ket{\Psi}_0  }{p_{\bm \kappa}}^2  \right) \ (\text{Definition variance}) \\
 &= 4\sum_{\bm  \kappa} p_{\bm \kappa} \left(\frac{\bra{\Psi}_0 \Pi_{\bm \kappa} (\Pi_{\bm \kappa}G \Pi_{\bm \kappa})^2 \Pi_{\bm \kappa} \ket{\Psi}_0  }{p_{\bm \kappa}}- \frac{\bra{\Psi}_0 \Pi_{\bm \kappa} (\Pi_{\bm \kappa}G\Pi_{\bm \kappa}) \Pi_{\bm \kappa} \ket{\Psi}_0  }{p_{\bm \kappa}}^2  \right) \ (\text{using } \Pi_{\bm \kappa}^2=\Pi_{\bm \kappa} \land  [\Pi_{\bm \kappa},G]=0) \\
 &=4\sum_{\bm  \kappa} p_{\bm \kappa} {\rm Var}_{\Psi_{\bm \kappa}} (G_{\bm \kappa}) \ (\text{Definition variance})
\end{align*}

\subsubsection{Derivation of last inequality of equation (\ref{equ: Boundprodprobdfs})} \label{Proof:equ: Boundprodprobdfs}
\begin{align*}
    \left(\frac{ m}{\lfloor \frac{m-1}{2}\rfloor}\right)^{-m} &\overset{(m-1) \ \text{even}}{=} \left(\frac{ m}{ \frac{m-1}{2}}\right)^{-m} = 2^{-m} \left( \frac{m-1}{m}\right)^m\leq 2^{-m}\\
    &\overset{(m-1) \ \text{odd}}{=} \left(\frac{ m}{ \frac{m-2}{2}}\right)^{-m} = 2^{-m} \left( \frac{m-2}{m}\right)^m\leq 2^{-m}
\end{align*}

\subsubsection{Deriviation equation (\ref{equ:fisher_info_term_ineq})}\label{Proof:equ:fisher_info_term_ineq} 
In the main text we find equation $4 p_{\bm\kappa}\mathrm{Var}_{\Psi_{\bm{\kappa}}}(G_{\bm\kappa}) \leq 2 \| G_{\bm \kappa}  \|_{\rm Spec}^2 \left(\sum_{i=2}^k p_i\right)$, which is proven in this section.
\begin{proof}
    As in the main text we set $p_{\bm  \kappa}:= \tr(\Psi_0 \Pi_{\bm  \kappa})$. We want to maximize \begin{equation}
    {F}_{\rm DFS}:= 4 p_{\bm  \kappa} {\rm Var}_{\Psi_{\bm  \kappa}} (G_\kappa) 
    \end{equation}
    We decompose the projector $\Pi_{\bm  \kappa}$ into two orthogonal projectors 
    \begin{equation}
        \Pi_{\bm  \kappa}=\Pi_{+} + \Pi_{-}, \ \Pi_{+} \Pi_{-}=0,
    \end{equation}
    where we set $p_{+}:=\tr(\Psi_0 \Pi_{+})$ and $p_{-}:=\tr(\Psi_0 \Pi_{-})$.
    We choose the decomposition such that the projector $\Pi_{+}$ projects $\ket{\Psi_0}$ onto the positive eigenspace of the generator $G$, where $\Pi_{-}$ projects $\ket{\Psi_0}$ onto the negative eigenspace of the generator $G$.
    Notice that $\mathrm{Var}_{\Psi_{\bm{\kappa}}}(G_{\bm \kappa}) \leq \mathrm{Var}_{\Psi_{\bm{\kappa}}}(g_{max} \Pi_+ + g_{min} \Pi_-)\leq \frac{||G_{\bm\kappa}||_{\rm spec}^2}{4}$, where $g_{max}$ and $g_{min}$ are the maximal and minimal eigenvalues of $G_{\bm \kappa}$ and without loss of generality we assume $g_{min}<0$.
    Using this observation, we find :
    \begin{align*}
        4 p_{\bm  \kappa}  {\rm Var}_{\Psi_{\bm  \kappa}} (G_\kappa) &\leq  4 p_{\bm  \kappa}  {\rm Var}_{\Psi_{\bm  \kappa}} (g_{max} \Pi_+ + g_{min} \Pi_-)\\
        &=4 p_{\bm \kappa} \left(\frac{\bra{\Psi_0} \Pi_{\bm \kappa} (g_{max}^2\Pi_++ g_{min}^2 \Pi_-) \Pi_{\bm \kappa} \ket{\Psi_0}  }{p_{\bm \kappa}}- \frac{\bra{\Psi_0} \Pi_{\bm \kappa} (g_{max} \Pi_+ + g_{min} \Pi_-) \Pi_{\bm \kappa} \ket{\Psi_0}  }{p_{\bm \kappa}}^2  \right) \\
        &= 4p_{\bm \kappa} \left( \frac{p_+ g_{max}^2+p_- g_{min}^2}{p_{\bm\kappa}} -\left( \frac{p_+ g_{max}+ p_- g_{min}}{p_{\bm \kappa}}\right)^2\right).
    \end{align*}
    The final expression for $ {F}_{\rm DFS}$ can be rewritten as 
    \begin{equation}
        {F}_{\rm DFS}= 4 \frac{p_+p_-}{p_++p_-} ||G_{\bm \kappa}||_{\rm spec}^2.
    \end{equation}
    The equation above can be upper bounded by $4 \ \text{min}\{p_+,p_-\} ||G_{\bm \kappa}||_{\rm spec}^2$.
    We introduce the probability $p_j$ as the probability to project onto one state of the $\text{DFS}_{\bm \kappa}$ labelled by $j$. We introduce $p_\perp$ as the probability to project onto any other state in the $\text{DFS}_{\bm \kappa}$, which is not labelled by $j$. Therefore, we have 
    \begin{equation}
        p_\kappa=p_j+p_\perp.
    \end{equation}
    Using this definition we rewrite $p_+$ and $p_-$ as
    \begin{equation}
        p_+=p_j+x, \ p_-=p_\perp-x,
    \end{equation}
    where $x$ is a parameter out of the interval $[0,p_\perp]$. So, we obtain 
    \begin{equation}
    {F}_{\rm DFS} \leq \underset{0\leq x \leq p_\perp}{\text{max}} 4 \frac{(p_j+x)(p_\perp-x)}{p_{\bm \kappa}} ||G_{\bm \kappa}||_{\rm spec}^2 \equiv {F}_{\rm DFS}^{\rm max}
    \end{equation}
    We perform the maximization via differentiating between the cases $p_j > \frac{p_{\bm \kappa}}{2}$ and $p_j \leq \frac{p_{\bm \kappa}}{2}$.
    \\
    \textbf{Case 1} $p_j > \frac{p_{\bm \kappa}}{2}$: \\
    In this case, more than half of the probability to project onto the DFS is concentrated on one state. Here we find that the maximum is obtained for $x_\text{max}=0$ and the maximum is 
    \begin{equation}
        {F}_{\rm DFS}^{\rm max}=4\left(p_j-\frac{p_j^2}{p_{\bm \kappa}} \right) ||G_{\bm \kappa}||_{\rm spec}^2.
    \end{equation}
    This expression, in turn, is maximized at its boundary for $p_j^{\rm max}=p_{\bm \kappa}/2$. Which directly leads to case 2.\\
    \textbf{Case 2} $p_j \leq \frac{p_{\bm \kappa}}{2}$: \\
     In this case, the probability to project onto a single state in the DFS contributes at most half of the overall probability to project onto the DFS. Here we find that the maximum is obtained for $x_\text{max}=p_j-\frac{p_\kappa}{2}$ and we find 
     \begin{equation}
         {F}_{\rm DFS}^{\rm max}=p_\kappa ||G_{\bm \kappa}||_{\rm spec}^2
     \end{equation}
    From above we know that $4 p_{\bm \kappa} {\rm Var}_{\Psi_\kappa} (G_{\bm \kappa})$ is maximal for $p_1 \leq p_{\bm \kappa}/2$. Therefore, we know that 
    \begin{equation} 
    4 p_{\bm \kappa} {\rm Var}_{\Psi_\kappa} (G_{\bm \kappa}) \leq 2 \sum_{i=2}^{k} p_i {\rm Var}_{\Psi_\kappa} (G_{\bm \kappa}),
    \end{equation}
    since $p_{\bm\kappa}=\sum_{i=1}^{k} p_i$. \\
    Moreover, we can bound the variance ${\rm Var}_{\Psi_\kappa} (G_{\bm \kappa}) \leq ||G_{\bm \kappa}||_{\rm spec}^2/4$: 
    In total, we have 
    \begin{equation}
        4 p_{\bm \kappa} {\rm Var}_{\Psi_\kappa} (G_{\bm \kappa}) \leq 2 \sum_{i=2}^{k} p_i \| G_{\bm \kappa}  \|_{\rm Spec}^2.
    \end{equation}
\end{proof}

\subsubsection{One DFS for minimal sized sensor networks}\label{Proof:m=n_DFS} 
Here, we show that for the case of $m=n$, i.e number of sensors is equal to the minimal number of sensors for which a DFS exists, that $\exists_1$ DFS, with $\text{dim(DFS)}=2$. \\
Here, we will invoke Lemma 1 and find that for any $\bm z, \bm z' \in \text{DFS}$, we have $D_H(\bm z, \bm z')=m=n$. Therefore, we have $\bm z'= - \bm z$. No other bit string can have Hamming distance $n$, so we find for each DFS: \begin{equation}
    \text{DFS}=\{\pm \bm z\}.
\end{equation}
Since $\bm z,- \bm z \in \text{DFS}$, we know that $\text{DFS}=\text{DFS}_{\bm \kappa=0}$. Moreover, let us assume that there exists another DFS having the same structure $\tilde{\text{DFS}}=\{\pm  \tilde{\bm z}\}$. From the definition of the DFS it follows that also 
\begin{equation}
    \tilde{\tilde{\text{DFS}}}=\{\pm \bm z, \pm  \tilde{\bm z} \},
\end{equation}
has to be a DFS. However, from $D_H(\bm z,  \tilde{\bm z})=n$ follows that $\tilde{\bm z}=- \bm z$. So, we have $\text{DFS}=\tilde{\text{DFS}}$. 

\end{document}